\documentclass[a4paper,12pt]{article}
\usepackage{authblk}
\usepackage[mathscr]{eucal}
\usepackage{amsmath,amsfonts,amssymb,amsthm}
\usepackage{times}
\usepackage{hyperref}

\newcommand{\ruth}[1]{{#1}}

 \newtheorem{prop}{Proposition}[section]
 
 \newtheorem{definition}[prop]{Definition}

\renewcommand{\tilde}{\widetilde}
\renewcommand{\hat}{\widehat}

\newcommand{\bref}[1]{\textbf{\ref{#1}}}

\newcommand{\im}{\mathop{\mathrm{Im}}}

\newcommand{\p}[1]{|#1|}
\newcommand{\gh}[1]{\mathrm{gh}(#1)}

\newcommand{\dv}{\mathrm{d_v}}
\newcommand{\dx}{\mathrm{d_X}}

\renewcommand{\d}{\partial}
\renewcommand{\dh}{\mathrm{d_h}}

\newcommand{\cF}{\mathcal{F}}

\newcommand{\tensor}{\otimes}

\renewcommand{\geq}{\,{\geqslant}\,}

\newcommand{\binner}[2]{%
  {\langle}\kern-4.15pt{\langle}#1{,}\,#2{\rangle}\kern-4.15pt{\rangle}}
\newcommand{\commut}[2]{[#1{,}\,#2]}

\newcommand{\qcommut}[2]{[#1{,}\,#2]_*}
\newcommand{\pb}[2]{\left\{{}#1{},{}#2{}\right\}}
\newcommand{\ab}[2]{\big(#1,#2\big)}

\newcommand{\half}{\mathchoice{%
    \ffrac{1}{2}}{\frac{1}{2}}{\frac{1}{2}}{\frac{1}{2}}}

\newcommand{\ffrac}[2]{\raisebox{.5pt}%
  {\footnotesize$\displaystyle\frac{#1}{#2}$}\kern1pt}

\newcommand{\dl}[1]{\mathchoice{\ffrac{\d}{\d #1}}{\frac{\d}{\d #1}}{\ffrac{\d}{\d #1}}{\ffrac{\d}{\d #1}}}

\newcommand{\manifold}[1]{\mathscr{#1}}

\newcommand{\manM}{\manifold{M}}

\newcommand{\Liealg}{\mathfrak} 
\newcommand{\algg}{\Liealg{g}}

\newcommand{\algA}{\mathcal{A}}

\newcommand{\cC}{\mathcal{C}}

\newcommand{\fR}{\mathbb{R}}

\newcommand{\fZ}{\mathbb{Z}}

 \def\cE{\mathcal{E}}
 \def\cG{\mathcal{G}}
 \def\cH{\mathcal{H}}
 \def\cI{\mathcal{I}}

 \def\cN{\mathcal{N}}
 
\def\tr{{\rm Tr}}

\usepackage{tikz-cd}
\usetikzlibrary{babel}

\numberwithin{equation}{section}
\newcommand{\comment}[1]{}
\usepackage{cleveref}
\newcommand{\cD}{\mathcal{D}}
\newtheorem{example}{Example}[section]

\usepackage{cleveref}
\usepackage{grffile}
\usepackage{etoolbox}
\usepackage{mathtools,indentfirst,cmap,mathtext}
\usepackage[left=1.5cm,right=1.5cm,
top=1cm,bottom=2cm,bindingoffset=0cm]{geometry}

\def\hL{\mathcal{L}}

\newcommand{\fu}{\cC^{\infty}}

\title{Notes on the $L_\infty$-approach to local gauge field theories}

\author{Maxim Grigoriev and Dmitry Rudinsky}
\affil{Lebedev Physical Institute,\protect\\
  Leninsky Ave. 53, 119991 Moscow, Russia \vspace{1em}
  \\
  Institute for Theoretical and Mathematical Physics,\protect\\
  Lomonosov Moscow State University, 119991 Moscow, Russia}  \vspace{1em}

\date{}  

\begin{document}
\newcommand{\linf}{\EuScript{H}}
\newcommand{\Id}{\text{id}}

\maketitle
\begin{abstract}
It is well known that a $Q$-manifold gives rise to an $L_\infty$-algebra structure on the tangent space at a fixed point of the homological vector field. From the field theory perspective this implies that the expansion of a classical Batalin-Vilkovisky (BV) formulation around a vacuum solution can be equivalently cast into the form of an $L_\infty$-algebra. In so doing, the BV symplectic structure determines a compatible cyclic structure on the $L_\infty$-algebra. Moreover, $L_\infty$ quasi-isomorphisms correspond to so-called equivalent reductions (also known as the elimination of generalized auxiliary fields) of the respective BV systems. Although at the formal level the relation is straightforward, its implementation in field theory requires some care because the underlying spaces become infinite-dimensional. In the case of local gauge theories,  the relevant spaces can be approximated by nearly finite-dimensional ones through employing the jet-bundle technique. In this work we study the $L_\infty$-counterpart of the jet-bundle BV approach and generalise it to a more flexible setup of so-called gauge PDEs. It turns out that in the latter case the underlying $L_\infty$-structure is analogous to that of Chern-Simons theory. In particular, the underlying linear space is a module over the space-time exterior algebra and the higher $L_\infty$-maps are multilinear.  Furthermore, a counterpart of the cyclic structure turns out to be degenerate and possibly nonlinear, and corresponds to a compatible presymplectic structure the gauge PDE, which is known to encode the BV symplectic structure and hence the full-scale Lagrangian BV formulation. Moreover, given a degenerate cyclic structure one can consistently relax the $L_\infty$-axioms in such a way that the formalism still describes non-topological models but involves only finitely-generated modules, as we illustrate in the example of Yang-Mills theory.
\end{abstract}

\tableofcontents
\section{Introduction}
Strongly homotopy associative algebras, now known as $A_\infty$-algebras, appeared first in~\cite{Stasheff:1963} in the context of topology. Their Lie algebra cousins, $L_{\infty}$-algebras, were first introduced and employed in~\cite{Zwiebach:1993ie}, as an algebraic structure underlying BV formulation of closed string field theory.  As algebraic objects, $L_{\infty}$-algebras were defined in~\cite{Lada:1993wc} and their study was initiated in~\cite{Lada:1994mn,Stasheff:1997fe,Barnich:1997ij}. 

The tight relation between $L_\infty$-algebras and the BV formulation of perturbative gauge theories was observed and employed already in~\cite{Zwiebach:1993ie} and then formalised in~\cite{Alexandrov:1995kv,Kontsevich:1997vb}  using the language of $Q$-manifolds. In brief, the tangent space at a fixed point of a homological vector field is naturally an $L_\infty$-algebra or, equivalently, formal pointed $Q$-manifolds are equivalent to $L_\infty$-algebras. 

In the field theory context, we are mostly interested in the Lagrangian gauge theories. At the classical level these are described in BV formalism as symplectic $Q$-manifolds. The symplectic structure on the BV configuration space corresponds to the cyclic structure of the associated $L_\infty$-algebra. The relation between the $L_\infty$ and BV formulation of general gauge theories is well understood by now, see e.g.~\cite{Hohm:2017pnh,Jurco:2018sby,Macrelli:2019afx,Arvanitakis:2019ald,Jurco:2020yyu,Ciric:2020hfj} as well as earlier relevant works \cite{Barnich:2004cr,Barnich:2005ru,Krotov:2006th,Losev2007FROMBI,Zeitlin:2007ttl,Mnev:2008sa}.

Despite being perturbative, the $L_\infty$-approach is quite natural from the QFT viewpoint as it makes explicit the underlying BRST first-quantized system and the perturbative interactions. In particular, the tree-level scattering amplitudes are captured by the minimal model of the $L_\infty$-algebra, see e.g. \cite{Kajiura:2003ax,Arvanitakis:2019ald,Macrelli:2019afx,Arvanitakis:2020rrk,Doubek:2020rbg}.  Moreover, at least formally, the passage to the quantum effective action can be understood in BV terms as the BV-pushforward (fiber-integration) or its equivalent in terms of loop-$L_\infty$ algebras~\cite{Zwiebach:1993ie,Krotov:2006th,Mnev:2006ch,Mnev:2008sa,Cattaneo:2009deo}.

In this work, after reviewing the relation  between the $L_\infty$ and BV approaches we elaborate on the interpretation of $L_\infty$ quasi-isomorphisms in terms of the geometry of the associated $Q$-manifolds. In particular, we show that passage to the minimal model can be understood in terms of the procedure, known as elimination of generalized auxiliary fields, giving a conventional definition of the classical $S$-matrix as the action evaluated on the perturbative solutions to the equations of motion~\cite{Arefeva:1974jv,Jevicki:1987ax}.

Our main focus is the $L_{\infty}$-approach to local gauge theories. The manifestly local  BV description of gauge field theories is well-known and operates in terms of the jet bundles. In this case the underlying BV configuration space is a suitable jet bundle equipped with the ghost degree and the BRST differential. The associated $L_{\infty}$-structure can be defined (at least locally in space-time) on sections of a natural graded vector bundle associated to a given solution (vacuum) of the underlying gauge theory. In so doing all the $L_{\infty}$-maps turn out to be multidifferential operators. The cyclic structure is determined by the BV symplectic structure.

A more flexible generalization of BV on jet bundles is known as gauge PDE or generalized AKSZ approach. In this approach a gauge theory is described as a $\fZ$-graded fiber bundle over $T[1](\text{spacetime})$, whose total space is equipped with a homological vector field $Q$ and possibly a compatible presymplectic structure. We identify an $L_\infty$-counterpart of this framework. The underlying $L_\infty$ structure is defined on the sections of a graded vector bundle over $T[1](\text{spacetime})$, and all the  $L_\infty$-maps, save for the differential, are multilinear over the space-time exterior algebra while the differential is the de Rham differential + a linear piece.

The paper is organized as follows: in Section \hyperref[sec:prel]{\textbf{2}} we recall the basic concepts of $L_\infty$-algebras and its application to field theories. In Section \hyperref[sec:qman]{\textbf{3}} we briefly review how the $L_\infty$-algebra formulation of field theories relates to formal $Q$-manifolds and BV formalism. Section \hyperref[sec:equivred]{\textbf{3}} is also devoted to describing a transition to the minimal model in geometrical terms of associated $Q$-manifolds and its relation to the procedure of elimination of generalized auxiliary fields. In Section \hyperref[sec:loc]{\textbf{4}} we elaborate the $L_\infty$-approach to local gauge theories. Conventions and notations are placed in the Appendix.

\section{Preliminaries}
\label{sec:prel}
\subsection{\texorpdfstring{$L_{\infty}$}{}-  algebras and gauge systems}

Let $\linf$ be a $\mathbb{Z}$-graded vector space $\linf:=\bigoplus_{k\in\mathbb{Z}}\linf_k$, with the degree denoted by $\gh{\cdot}$, which is equipped with a collection $l_i$, $i=1,2,3,...$ of multilinear maps:
\begin{equation}\label{formula1}
	l_i:\linf ^{\times i}\rightarrow \linf\,, \qquad \gh{l_i}=1\,.
\end{equation}
These maps are assumed to be graded symmetric and satisfying the generalized graded Jacobi identities. In particular, the first map is the differential $l_1:=\cD$, while the second map $l_2$ defines an algebra structure that satisfies a version of Jacobi identity modulo $\cD$-exact terms. For instance if only $\cD$ and $l_2$ are nonvanishing, $L_\infty$-algebra is just a differential graded Lie algebra (note that in the present exposition the degree is suspended by $1$ with respect to the standard choice for differential graded Lie algebras). For the first three maps the generalized Jacobi identity reads as
\begin{equation}\label{c11}
	\cD(\cD(\phi))=0,
\end{equation}
\begin{equation}\label{c12}
	\cD(l_2(\phi_1,\phi_2))+l_2(\cD(\phi_1),\phi_2)+(-1)^{|\phi_1|}l_2(\phi_1,\cD(\phi_2))=0,
\end{equation}
\begin{align}\label{c13}
	l_3(l_2(\phi_1,\phi_2),\phi_3)+(-1)^{|\phi_2||\phi_3|}l_3(l_2(\phi_1,\phi_3),\phi_2)+(-1)^{|\phi_1|(|\phi_2|+|\phi_3|)}l_3(l_2(\phi_2,\phi_3),\phi_1)\\
	+\cD(l_3(\phi_1,\phi_2,\phi_3))+l_3(\cD(\phi_1),\phi_2,\phi_3)+(-1)^{|\phi_1|}l_3(\phi_1,\cD(\phi_2),\phi_3)\nonumber\\
	+(-1)^{|\phi_1|+|\phi_2|}l_3(\phi_1,\phi_2,\cD(\phi_3))=0,\nonumber
\end{align}
while the general formula can be found in e.g.~\cite{Jurco:2018sby,Macrelli:2019afx}. The first of the above relations implies that $\cD$ is a differential and hence $\linf$ is a cohomological complex. The second says that $\cD$ is a derivation of $l_2$, while 
the third implies that $l_2$ defines a Lie algebra structure in the cohomology of $\cD$.

The data of the $L_\infty$ algebra encodes a gauge system or, more precisely, a gauge system formulated in terms of BV-BRST formalism, i.e. where in addition to equations of motion one specifies gauge generators, gauge for gauge generators and higher gauge generators. An element $\varphi\in \linf_0$ is called a solution if 
\begin{equation}\label{formula5}
	\sum_{i\geq 1}\frac{1}{i!}l_i(\varphi,\dots,\varphi)=0,
\end{equation}
so that the above relations can be regarded as equations of motion. It is easy to check that the following transformation
\begin{equation}\label{fo6}
	\delta_{\epsilon}\varphi:=\sum_{i\geq 0}\frac{1}{i!}l_{i+1}(\epsilon,\varphi,\dots,\varphi)\,,
\end{equation}
where $\epsilon$ is a generic element of $\linf_{-1}$, which is identified as a gauge parameter, is a symmetry of\eqref{formula5}.  Analogously, one defines gauge for gauge symmetries
\begin{equation}
	\delta_{\epsilon_{-1}}\epsilon:=\sum_{i\geq 0}\frac{1}{i!}l_{i+1}(\epsilon_{-1},\varphi,\dots,\varphi)\,, 
\end{equation}
where $\epsilon_{-1}$ is a generic element of $\linf_{-2}$, regarded as a gauge for gauge parameter.  Generalization to higher gauge symmetries is straightforward.

It is important to note that the above representation of equations of motion and gauge transformations is not fully-applicable to systems containing fermionic fields and/or fermionic gauge parameters.  To incorporate fermions one equips $\linf$ with an addition structure --  $\mathbb{Z}_2$ Grassman degree  which defines a graded symmetry properties of the maps $l_i$. In what follows, we assume that $\p{-}$ denotes the Grassmann degree. In particular, if fermions are not present the Grassmann degree is simply the ghost degree modulo 2. However, the subtlety is that if fermions are present it is difficult to interpret $\varphi\in \linf_0$ as a classical field configuration because components of $\varphi$ are real numbers (recall that $\linf$ is a linear space) rather then anticommuting generators of the Grassmann algebra. Recall, that at the classical level fermionic fields are described by anticommuting variables. Nevertheless,  at the purely algebraic level all the information of the underlying gauge system is captured by the $L_\infty$-maps. Moreover, fermionic  fields naturally appear if one passes to the associated BV-BRST description, as we will recall in the next Section.

In applications, we are mostly interested in Lagrangian systems which, in this approach, correspond to cyclic $L_\infty$ algebras.  By definition $(\linf,l_i)$ is cyclic if it is equipped with a  bilinear form:
\begin{equation}\label{formula6}
	\kappa:\linf\tensor \linf\rightarrow \mathbb{R}\,,\quad gh(\kappa)=-1\,, \quad  \kappa(\phi_1,\phi_2)=(-1)^{(|\phi_1|+1)(|\phi_2|+1)}\kappa(\phi_2,\phi_1)\,,
\end{equation} 
compatible with the $L_\infty$-structure in the following way:
\begin{equation}\label{3}
\kappa(\phi_1,l_i(\phi_2,\dots,\phi_{i+1}))=(-1)^{|\phi_1||\phi_2|}\kappa(\phi_2,l_i(\phi_1,\phi_3,\dots,\phi_{i+1}))\,.
\end{equation}
Note that $\kappa(\phi_1,\phi_2)\neq 0$ iff the degree of $\phi_1$ is equal to $1-|\phi_2|$, therefore $(-1)^{(|\phi_1|+1)(|\phi_2|+1)}=1$ and then $\kappa$ is in fact symmetric, i.e. $\kappa_{AB}=\kappa_{BA}$, where $\kappa(e_A,e_B)=\kappa_{AB}$ is a matrix of bilinear form and $e_A$ denotes a basis in $\linf$. Note that one can equivalently regard $\kappa$ as a symplectic structure of degree $-1$. Such an interpretation appears useful in discussing the relation to BV formalism.

Given a cyclic $L_{\infty}$ algebra $\linf$ one can define a formal power series $S[\varphi]$ on $\linf_0$ as follows
\begin{equation}\label{formula9}
	S[\varphi]:=\sum_{i\geq 1}\frac{1}{(i+1)!}\kappa(\varphi,l_i(\varphi,\dots,\varphi))\,.
\end{equation}
It is interpreted as an action functional of the underlying gauge theory. In particular, \eqref{formula5} describes the stationary points of \eqref{formula9} while \eqref{fo6} describes its gauge symmetries.

\begin{example}\label{ex:CS}\normalfont  As an instructive example let us consider the $L_\infty$-formulation of Chern-Simons theory, see e.g. \cite{Hohm:2017pnh,Ciric:2020hfj}. Suppose we have a compact orientable manifold $X$, $dim(X)=3$ and denote by $(\Omega^{\bullet}(X),d)$ its de Rham complex. Let $\algg$ be a real Lie algebra equipped with an invariant inner product. We take as $\linf=\Omega^{\bullet}(X)\tensor \algg[1]$, i.e. differential forms with values in $\algg[1]$. For instance, an element $\varphi\in \linf_0$ can be written as $\varphi=dx^{\mu}\varphi^a_{\mu}t_\alpha$, where $t_\alpha$ denote a basis in $\algg[1]$.

The $L_\infty$ maps are defined as follows:
\begin{equation}\label{e11}
	l_1(\phi):=\mathsf{d}\phi, \qquad 
	l_2(\phi_1,\phi_2):=(-1)^{|\phi_1|}[\phi_1,\phi_2]\,,
\end{equation}
where  $[\phi_1,\phi_2]$ denotes a Lie bracket in $\algg$ extended to $\linf$ by linearity over $\Omega^{\bullet}(X)$.  The $L_\infty$ relations are satisfied thanks to $d^2=0$, graded Leibnitz rule for $d$, and the Jacobi identity for $\algg$.
The cyclic structure on $\linf$ is defined as
\begin{equation}\label{CHS1}
	\kappa(\phi_1,\phi_2):=(-1)^{|\phi_1|}\int_X \tr (\phi_1 \wedge \phi_2)\,,
\end{equation}
where $\tr$ denotes an invariant bilinear product in $\algg$ extended to $\linf$ by linearity over $\Omega^\bullet(X)$. It is symmetric for elements with degree $|\phi_1|+|\phi_2|=1$ and equal to zero for others, i.e. (\ref{CHS1}) has the same symmetry as (\ref{formula6}).

In this way we have arrived at the cyclic $L_{\infty}$-algebra $\left(\linf,l_i,\kappa\right)$, where $l_i$ are given by \eqref{e11}. The associated classical action \eqref{formula9} takes the following form:
\begin{equation}
	S_{CS}=\frac{1}{2}\int_{X}<\varphi,\mathsf{d}\varphi+\frac{1}{3}[\varphi,\varphi]>.
\end{equation}	
and is precisely the action of Chern-Simons theory. The equations of motion \eqref{formula5} read as:
\begin{equation}
	\mathsf{d}\varphi+\frac{1}{2}[\varphi,\varphi]=0\,,
\end{equation}
while the gauge transformations are as follows:
\begin{equation}
	\delta_{\epsilon}\varphi=\mathsf{d}\epsilon+[\varphi,\epsilon],
\end{equation}
with $\epsilon \in \linf_{-1}=\Omega^0(X)\otimes \algg[1]$. 

Let us note that the assignment of ghost degree that we employed in the above example is not so common in the literature, though~\cite{Barnich:2004cr}. An alternative is to suspend the degree by 1 so that the ghost degree is simply the form degree, see e.g. ~\cite{Jurco:2019bvp,Jurco:2020yyu,Ciric:2020hfj}. Although our choice is slightly unnatural in this example it is more useful for explaining the relation with the BV formulation.
\end{example}

\subsection{$L_\infty$-morphisms and minimal models}

Natural morphisms of $L_\infty$ algebras are those that preserve the structure up to homotopy. What is important is that $L_{\infty}$ - morphism is not a linear map in general. More specifically, $\chi:(\linf,l_i)\rightarrow (\linf^{\prime}, l_i^{\prime}) $ between two $L_{\infty}$-algebras $(\linf,l_i)$ and $(\linf^{\prime},l^{\prime}_i)$ is a set of multilinear graded antisymmetric maps $\chi_i:\linf^{\times i}\rightarrow \linf^{\prime}$ of degree $0$. The conditions that these maps satisfy look rather cumbersome but, as we are going to see in the next section,  can be compactly written in the language of $Q$-manifolds so that we refrain from giving their explicit form here and refer instead to e.g.~\cite{Kajiura:2003ax,Macrelli:2019afx}.

The $L_{\infty}$-morphism $\chi$ is called \textit{quasi-isomorphism} if the map $\chi_1:\linf\rightarrow \linf^{\prime}$ is a quasi-isomorphism of the underlying complexes, i.e. $H^{\bullet}_{\mathcal{D}}(\linf)\cong H^{\bullet}_{\mathcal{D}^{\prime}}(\linf^{\prime})$.
It is important to note that just like complexes of vector spaces quasi-isomorphic $L_{\infty}$-algebras are not necessarily isomorphic as linear spaces and among quasi-isomorphic $L_{\infty}$-algebras there is a minimal one, called minimal-model~\cite{Sullivan:1977fk,Kadeishvili1980}. More precisely:
\begin{definition}
A minimal model of $(\linf,l_i)$ is an $L_{\infty}$-algebra $(\linf^{\prime},l^{\prime}_{i})$ such that $(\linf,l_i)$ and $(\linf^{\prime},l^{\prime}_{i})$ are quasi-isomorphic and 
$l^{\prime}_{1}=0$.
\end{definition}
In particular, as a linear space $\linf^{\prime}$ is isomorphic to the cohomology $H^{\bullet}_{\mathcal{D}}(\linf)$.

Let us briefly recall how the minimal model, i.e. maps $l^{\prime}_{i\ge 2}$ defined on $\linf^\prime= H^{\bullet}_{\mathcal{D}}(\linf)$,  can be found for a given $L_{\infty}$-algebra.  To this end let us define the following three linear maps: the first two are the embedding $\mathbf{e}:\linf^\prime \hookrightarrow \linf$ (recall that $\linf^\prime \cong H^{\bullet}_{\mathcal{D}}(\linf)$) and the projection $\mathbf{p}:\linf\twoheadrightarrow \linf^\prime$, which are both of degree $0$ and satisfy $\textbf{p}\circ\textbf{e}=\Id$, where $\Id$ is an identity map on $\linf^\prime$. The third map is a \textit{contracting homotopy}  $\textbf{h}:\linf\rightarrow \linf$ that has degree $-1$ and satisfies:
\begin{equation}\label{MM4}
	\text{id}_\linf-\textbf{e}\circ\textbf{p}=\textbf{h}\circ \mathcal{D}+\mathcal{D}\circ\textbf{h}\,.
\end{equation}
The quasi-isomorphism $\chi$ between $(\linf,l_i)$ and $(\linf^{\prime},l^{\prime}_{i\ge 2})$  is now determined by the following maps (where we only give explicitly some of them)
\begin{equation}\label{MM1}
\chi_1(\phi'_1):=\textbf{e}(\phi'_1),
\end{equation}
\begin{equation}\label{MM2}
	\chi_2(\phi'_1,\phi'_2):=-(\textbf{h}\circ l_2)(\chi_1(\phi'_1),\chi_1(\phi'_2))
\end{equation}
\begin{align}
    \chi_3(\phi'_1.\phi'_2,\phi'_3):=-\sum_{Permutations}\frac{1}{2!}(\textbf{h}\circ l_2)\left((\chi_{1}(\phi'_1),\chi_{2}(\phi'_{2},\phi'_{3}))+(\chi_2(\phi'_1,\phi'_2),\chi_1(\phi'_3)\right)-\\
   -\sum_{Permutations}\frac{1}{3!}(\textbf{h}\circ l_3)(\chi_1(\phi'_1),\chi_2(\phi'_2),\chi_3(\phi'_3))\nonumber
\end{align}
where $\phi'_1, \phi'_2, \phi'_3\in \linf^{\prime}$.
The $L_{\infty}$-structure on $\linf^{\prime}$ is determined as
\begin{equation}\label{MM3}
	l'_1(\phi'_1):=0,\qquad 
	l'_2(\phi'_1,\phi'_2):=(\textbf{p}\circ l_2)(\chi_1(\phi'_1),\chi_1(\phi'_2)),\qquad \ldots
\end{equation}
where, again, we only gave the first two maps. All higher maps are defined recursively in a similar way, see e.g~\cite{Macrelli:2019afx}.

\section{Q-manifolds and \texorpdfstring{$L_{\infty}$}{}-algebras}
\label{sec:qman}

In this section we recall the basics of $Q$-manifolds and the relation between $Q$-manifolds and $L_\infty$-algebras.
\subsection{$Q$-manifolds}
\begin{definition}
	A $Q$-manifold is a $\mathbb{Z}$-graded supermanifold endowed with an odd vector field $Q$ of degree $1$ such that $[Q,Q]=2Q^2=0$. Such a vector field is called homological.
\end{definition}
Let us give here a couple of standard examples:
\begin{example}\label{ex1}
\normalfont Given a (graded) manifold $M$, the shifted tangent bundle $T[1]M$ is naturally a $Q$-manifold, with the $Q$-structure being the de Rham differential $d_M=\theta^i\frac{\partial}{\partial x^i}$ where $x^i,i=1,\cdots,d$ are coordinates on $M$ and $\theta^i$ are coordinates on the fibers of $T[1]M$.
\end{example}
\begin{example}
\normalfont Let $\algg$ be a Lie algebra. Consider $\algg[1]$ as $\mathbb{Z}$-graded supermanifold. Consider a $Q$-field $Q=-\frac{1}{2}f^{\alpha}_{\beta\gamma}c^{\beta}c^{\gamma}\frac{\partial}{\partial c^{\alpha}}$,
where $c^\alpha$ are linear coordinates on $\algg[1]$ and  $f^{\alpha}_{\beta\gamma}$ the structure constants of $\algg$. The nilpotency of $Q$-field is equivalent to the Jacobi identity in $\algg$.
\end{example}

Natural maps of $Q$-manifolds, sometimes called $Q$-maps, are maps of the underlying graded supermanifolds that respect the $Q$ structure, i.e. $\phi:M_1 \to M_2$ satisfies $\phi^*\circ Q_2=Q_1 \circ  \phi^*$. The $Q$-map condition encodes the properties of $L_{
\infty}$-morphism~\cite{Kontsevich:1997vb}. For more details on morphisms of graded manifolds see e.g.~\cite{Voronov:2019mav}. 

$Q$-manifolds give an alternative way to describe gauge systems. Roughly speaking $Q$-manifold is a geometrical object that corresponds to a BV formulation, where a symplectic structure is not specified or even does not exist. From the field theory perspective this corresponds to describing gauge systems at the level of equations of motion. Let us recall how the data of a gauge system is encoded in a given $Q$-manifold. For future convenience we do it in geometrical terms following~\cite{Grigoriev:2019ojp}. An explicit description in terms of coordinates was initially in~\cite{Barnich:2004cr,Lyakhovich:2004xd}.

Solutions of the underlying gauge system, described by $(M,Q)$, can be identified as points of the body of $M$, where $Q$ vanishes. These can be identified as $Q$ maps from a trivial $Q$-manifold $(pt,0)$ to $(M,Q)$. In other words $\sigma:pt\rightarrow M$  is a solution if 
\begin{equation}\label{m1}
 	\sigma^*\circ Q=0\,.
\end{equation}
In other words $\sigma(pt)$ is a fixed point of $Q$. Gauge transformation of a given map $\sigma$ can be defined as 
\begin{equation}\label{m3}	\delta_{\epsilon_\sigma}\sigma^*=\epsilon_\sigma^*\circ Q,
\end{equation}
where the gauge parameter map $\epsilon_\sigma:C^{\infty}(M)\rightarrow \mathbb{R}^1$ has degree $-1$ and satisfies
\begin{equation}\label{m4}
	\epsilon_\sigma^*(fg):=\epsilon_\sigma^*(f)\sigma^*(g)+(-1)^{|f|}\sigma^*(f)\epsilon_\sigma^*(g),
\end{equation}
for all $f,g \in C^{\infty}(M)$. In fact it is enough to consider gauge parameters of the following form: $\xi^*_\sigma=\sigma^*\circ Y$, where $Y$ is a vector field on $M$ of degree $-1$ so that \eqref{m4} is satisfied by construction. It is easy to check that~\eqref{m3} defines an infinitesimal deformation of the solution $\sigma$. Gauge for gauge symmetries can be defined in analogously.

\subsection{$Q$-manifolds and $L_\infty$-algebras}

Given an $L_{\infty}$-algebra $(\linf,l_i)$ there is a canonically associated formal pointed $Q$-manifold that encodes all the data of $(\linf,l_i)$.  The construction is as follows: any $\phi\in \linf$ can be represented in the form
\begin{equation}
	\phi=\phi^Ae_A,
\end{equation}
where $\phi^A\in\mathbb{R}$ and $e_A$ denote a basis in $\linf$. Let $\psi^A$ denote a dual basis in the graded dual of $\cH^*$, i.e. $	gh(\psi^A):=-gh(e_A)$ and $ |\psi^A|:=|e_A|$, where $ |\ | $ denotes Grassmann parity.
Now let us consider the (suitable completion of the) graded commutative algebra generated by $\psi^A$ as an algebra of functions on a  graded supermanifold $M_\linf$, which is a graded supermanifold associated with $\linf$. In so doing, $\psi^A$ are interpreted as homogeneous linear coordinates on $M_\linf$. More invariantly, the algebra of functions on $M_\linf$ is the graded-symmetric tensor algebra of $\cH^*$.

Now we describe a $Q$-structure on $M_\linf$. Let us consider functions on $M_\linf$ with values in $\linf$, i.e. $C^{
\infty}(M_\linf)\tensor \linf$. This space has  a distinguished element $\Psi$, often called string field,
which can be defined in local coordinates 
as follows \footnote{Here by some abuse of notations we do not distinguish the degree in $\linf$ and $\linf^*$ and the total degree in $C^{
\infty}(M_\linf)\bigotimes \linf$. Though sometimes it can be useful to distinguish the degrees.}
\begin{equation}\label{formula24}
    	\Psi:=\psi^Ae_A, \quad gh(\Psi):=gh(\psi^A)+gh(e_A)=0,\quad |\Psi|=0.
\end{equation}
It is a linear function on $ M_\linf $ with values in $\linf$ \footnote{There is an alternative convention such that $gh(\Psi)=1$ and this corresponds to a different choice of grading of maps \eqref{formula1}, see e.g.~\cite{Grigoriev:2006tt,Zeitlin:2007vv,Zeitlin:2007ttl}.}. Because  $\Psi$ is associated with a graded space, we have the following decomposition
\begin{equation}\label{for26}
	\Psi=\sum_k \Psi^{(k)}=\sum_k \psi^{A_k}e_{A_k}\,, \qquad \gh{\psi^{A_k}}=k\,.
\end{equation}
In particular, $\Psi^{(0)}$ contains physical (=ghost degree $0$) fields.

Now the $Q$-structure on $M_\linf$ can be defined as follows:
\begin{equation}\label{c31}
	Q(\Psi):=\tilde{l}_1(\Psi)+\frac{1}{2}\tilde{l}_2(\Psi,\Psi)+\frac{1}{3!}\tilde{l}_3(\Psi,\Psi,\Psi)+\dots\,.
\end{equation}
where $\tilde{l}_i$ is an extension of the $L_\infty$-maps $l_i$ from $\linf$ to 
$C^{\infty}(M_\linf)\tensor\linf$ by linearity. We define this extension as follows 
$\tilde{l}_i(\Psi,\dots,\Psi)=\psi^{B_1}\dots\psi^{B_i}U^C_{B_1\dots B_i}e_C$, where $l_i(e_{B_1},\dots,e_{B_i})=U^C_{B_1\dots B_i}e_C$. 
This defines the action of $Q$ on coordinates and hence on generic functions on $M_\linf$. It is instructive to write \eqref{c31} in terms of components:
\begin{equation}\label{MM5}
	Q^A(\psi)e_A=\psi^A\mathcal{D}^B_{A}e_B+\frac{1}{2}\psi^A\psi^B U^C_{AB}e_C+\dots.
\end{equation}
The nilpotency condition $Q^2(\Psi)=0$ is satisfied as a consequence of the generalized Jacobi identity.

By construction, the above $Q$ has degree $1$, doesn't have  a constant term, and is nilpotent on account of the $L_\infty$-conditions. In fact $Q^2=0$ is often considered as a compact way to write the underlying $L_\infty$-conditions. It is also clear that $Q$-structure determined by \eqref{c31} encodes the entire structure of $(\linf,l_i)$. In other words formal pointed $Q$-manifolds  (with vanishing constant term in $Q$) and $L_\infty$-algebras are equivalent objects. 

However, general $Q$-manifolds 
can not be obtained this way. Nevertheless,  given a smooth $Q$-manifold, the structure of $L_\infty$-algebra is induced on its tangent space at any fixed point of $Q$~\cite{Alexandrov:1995kv}. More precisely, let $p\in M$ belongs to the zero locus of $(M,Q)$, i.e. $Q|_{p}=0$. Choose local coordinate system $\psi^A$ centered in $p$ and consider the Taylor expansion of the $Q$-field:
\begin{equation}\label{l1}
	Q^A(\psi)=\sum_{i\geq 1}\frac{1}{i!}Q^A_{B_1\cdots B_i}\psi^{B_1}\cdots\psi^{B_i}\,.
\end{equation}  
Note that the constant term vanishes as $p$ belongs to the zero locus. Coefficients $Q_i$ in \eqref{l1}) are graded-symmetric and induce the following maps 
\begin{equation}
	l_i:T_p^{\otimes i}M\rightarrow T_pM,
\end{equation}
where $l_i(e_{A_1},\dots,e_{A_i})=Q^B_{A_1,\dots,A_i}e_B$ and $e_A$ are basis elements in $T_pM$. So we have an $L_{\infty}$-algebra structure on the tangent space of $M$ at $p$ and $L_{\infty}$-conditions follow from $Q^2=0$. In this way one can also endow $TM$ pulled back to the zero-locus of $Q$ with the structure of an $L_{\infty}$-bundle.

It is important to note that maps $l_i$ of the above $L_\infty$-algebra on $T_pM$ depend on the choice of coordinate system $\psi^A$ for $i>1$. Because changes of coordinates result in $L_\infty$-isomorphisms the equivalence class of $L_\infty$ structure on $T_pM$ is well-defined. As an alternative, the $L_\infty$ structure can be made coordinate-independent via introducing an auxiliary torsion-free affine connection in $TM$ and lifting $Q$ to a section of the bundle, whose fiber at $p\in M$ is the Lie algebra of formal vector fields on $T_pM$, giving a globally well-defined  curved $L_\infty$-structure on $\Gamma(TM)$. The lift can be performed in a simple way using the technique originating from Fedosov quantization and is described in Appendix~\bref{app:Fedosov}. This gives an alternative derivation of the 
curved $L_\infty$-structure on $\Gamma(TM)$ constructed in~\cite{Mehta_2015,Seol:2021tol}, which in turn originates from ~\cite{kapranov1999rozansky}. For a recent discussion of these structures in the context of field theory see~\cite{Chatzistavrakidis:2023otk}.
Let us also note that a coordinate free way to define  an $L_\infty$-maps determined by $Q$ is to employ the higher derived brackets~\cite{Voronov_2005}, as we do in a more general setup in Section~\bref{sec:gPDE-Li}.

\subsection{Symplectic/cyclic structures and BV formalism}

From the geometrical perspective BV system is a symplectic $Q$-manifold, i.e. a $Q$-manifold $M$ equipped with a $Q$-invariant symplectic structure $\omega$ of degree $-1$.  It follows from $L_Q\omega=0$, where $L$ denotes Lie derivative, that there exists $S_{BV}$ such that \begin{equation}\label{c45}
	i_Q\omega+dS_{BV}=0\,.
\end{equation}
In other words, $Q$ is a Hamiltonian vector field with Hamiltonian $S_{BV}$. If $V_1,V_2$ are Hamiltonian vector fields with Hamiltonians $H_1,H_2$ respectively. The odd Poisson bracket  of the Hamiltonians is defined as
\begin{equation}
\ab{H_1}{H_2}=\omega(V_1,V_2)\,.
\end{equation}
If $\omega$ is invertible (as in our setup) the Poisson bracket is defined for any functions. For $S_{BV}$ one has
\begin{equation}
	(S_{BV},S_{BV})=\omega(Q,Q)=0.
\end{equation}
This is the classical master equation of BV formalism and it ensures that 
$Q=\ab{S_{BV}}{-}$ is nilpotent.

Let us now turn to the extra structure induced by $\omega$ on $T_pM$ for $p$
belonging to the zero locus. It turns out that in addition to the induced $L_\infty$-structure, $T_pM$ is equipped with a symplectic form of degree $-1$. To discuss compatibility between the $L_\infty$ and the symplectic structure it is convenient to employ Darboux coordinates centered in $p$ so that  $\omega$ takes the form 
\begin{equation}\label{s1}
    \omega=\frac{1}{2}d\psi^A d\psi^B \omega_{AB}\,,
\end{equation}
with $\omega_{AB}$ constant and has the following symmetry property: $\omega_{AB}=(-1)^{(|A|+1)(|B|+1)}\omega_{BA}$. Because $\gh{\omega}=1$ it is nonvanishing only if $|A|=|B|+1~mod~2$ and hence $\omega_{AB}=\omega_{BA}$ so that in our setup $\omega_{AB}$ has the same symmetry property as a cyclic product $\kappa_{AB}$. Furthermore, $L_{Q}\omega=0$ implies the cyclic property of the symplectic form on $T_pM$.

Other way around, given a cyclic $L_\infty$-algebra $(\linf,l_i,\kappa)$ it determines a unique symplectic $Q$-manifold. Indeed, $\kappa$ defines the constant symplectic structure on the formal $Q$-manifold $(M_\linf,Q)$ and the cyclicity condition implies that it is $Q$-invariant. Furthermore, the BV master action can be written explicitly in terms  of the structure maps of the underlying $L_\infty$-algebra. More precisely, in terms of the string field $\Psi$ the BV-action takes the form:
\begin{equation}\label{c44}
	S_{BV}=\sum_{n\geq 2}\frac{1}{n!}\omega_{B_1 C}Q^C_{B_{2\dots}B_n}\psi^{B_1}\dots\psi^{B_n},
\end{equation}
Note that the BV-action (\ref{c44}) has the same form as an action (\ref{formula9}), however the action (\ref{formula9}) contains only elements of degree zero while the BV-action includes elements of different degrees (fields, ghosts, and their conjugate antifields).

\subsection{Equivalent reduction  of \texorpdfstring{$Q$}{}-manifolds and minimal models}
\label{sec:equivred}

We now recall the notion of equivalent reductions of $Q$ manifolds and show that in the case of formal $Q$-manifolds this corresponds to quasi-isomorphic reduction of the respective $L_\infty$-algebras.  It is useful to introduce:
\begin{definition}
A contractible $Q$-manifold is that of the form $(T[1]V,d_V)$, where 
$V$ is an $\fZ$-graded vector space understood as a graded supermanifold and $d_V$ is the de Rham differential on $V$
\end{definition}
The formal $Q$-manifold version of this definition was in~\cite{Kontsevich:1997vb}. The coordinate version of this definition was, somewhat implicitly, employed in the BRST approach to field theory, see e.g.~\cite{Barnich:1995ap,Brandt:1996mh,Barnich:2004cr}.  

Let us also note that the direct product of $Q$-manifolds $(M_1,Q_1)$ and $(M_2,Q_2)$ is naturally a $Q$-manifold, whose $Q$-structure is uniquely determined by the property $Q(fg)=Q_1(f)g+(-1)^{|f|}fQ_2(g)$ for any functions $f\in \fu(M_1)$ and $g\in \fu(M_2)$ (here by some abuse of notations we identify functions on $M_{1,2}$ and their pull-backs to the direct product). We also need a notion of a $Q$-bundle which is a locally-trivial fiber-bundle where the total space and the base are $Q$-manifolds and the projection is a $Q$-map~\cite{Kotov:2007nr} (see also \cite{Mehta:2007rgt,Voronov:2009nr}). 

\begin{definition}
	Given two $Q$-manifolds $(M_1,Q_1)$ and $(M_2,Q_2)$, then $(M_2,Q_2)$ is called an equivalent reduction of $(M_1,Q_1)$ if $(M_1,Q_1)$ is a locally-trivial $Q$-bundle over 
	$(M_2,Q_2)$ which admits a global $Q$-section and whose fibers are contractible $Q$-manifolds. 
\end{definition}
Note that locally trivial $Q$-bundle is locally isomorphic to a direct product of the base and the typical fiber, where the direct product means the direct product of $Q$-manifolds. 

The above definition is taken from~\cite{Grigoriev:2019ojp} but 
it originates from the coordinate definitions  employed in the BRST approach to field theory, see e.g.~\cite{Barnich:1995ap,Brandt:1996mh} and~\cite{Henneaux:1990ua,Dresse:1990dj} for earlier Lagrangian version. In particular, the coordinate local version of the equivalent reduction was in~\cite{Barnich:2004cr}. The notion of equivalence generated by the above equivalent reduction is very restrictive and one might be interested in a more general homotopy equivalence of $Q$-manifolds, see e.g. recent discussion of the equivalence of BV systems~\cite{Simao:2021xgw}. Nevertheless, the above notion is sufficient for our applications in local field theory.

A  practical way to find an equivalent reductions is as follows~\cite{Barnich:2004cr}: let $w^a$ be independent functions such that 
$w^a$ and $Qw^a$ are also independent. Then, at least locally, the submanifold singled out by $w^a=0$ and $Qw^a=0$ is an equivalent reduction. Note that $Q$ is tangent to the submanifold and hence makes it into a $Q$-manifold. In the infinite-dimensional case one should also make sure that the respective bundle is indeed trivial (constructing the trivialization may take outside the functional class). 
\begin{prop}
Let $(M_\linf,Q)$ be a formal $Q$-manifold associated with an $L_{\infty}$-algebra $(\linf,l_i)$, then there exists an equivalent reduction $(N,Q|_N)$, which is a $Q$-manifold associated with a minimal model of $(\linf,l_i)$.
\end{prop} 
\begin{proof}
Let us fix a decomposition $\linf=\mathcal{E}\oplus \cF\oplus Im(\mathcal{D})$, where $\mathcal{E}\oplus Im(\mathcal{D})=Ker(\mathcal{D})$ and $\mathcal{E}\cong H^{\bullet}_{\mathcal{D}}(\linf)$. Consider a vector space of linear functions on $\manM_\linf$. These are $1:1$ with linear functionals on $\linf$ and hence we can consider a subspace $W$ of functions that corresponds to linear functionals vanishing on $\cE\oplus \cF$. Then we consider an ideal generated by $W$ and $QW$ and take a quotient of all functions on $\manM_\linf$ by the ideal. This quotient can be identified with functions on $\cE$ (because  for nonvanishing $w$ the linear piece of  $Qw$ is also nonvanishing, i.e. $Qw^a=v^a+o(\psi^2)$). This way we arrive at $Q$-manifold whose underlying linear space is $\cE\cong H^\bullet(\cD)$, the embedding is a $Q$-morphism and hence defines an $L_\infty$-morphism of the underlying $L_\infty$ algebras.
\end{proof}
The above proof can be made even more explicit by picking a basis in $\linf$ adapted to the decomposition $\linf=\mathcal{E}\oplus \cF\oplus Im(\mathcal{D})$ and then taking a dual basis $w^a$ in $(\im{\cD})^*$ seen as linear functions on $\manM_\linf$. The minimal model submanifold is then the zero locus of $w^a$ and $Qw^a$. In these terms it is easy to obtain an explicit formulas for the structure maps of the reduced algebra.

The above statement has a straightforward generalisation. Let $\cG\subset \im{\cD}$ be a subspace of $\im{\cD}$. Taking as $w^\alpha$ the basis elements of $\cG^*$ seen as linear functions of $\manM_\linf$ (this implies picking a decomposition $\linf=\cG\oplus \text{"the rest"}$) immediately defines the equivalent  reduction to a $Q$-submanifold singled out by $w^\alpha=0$ and $Qw^\alpha=0$. The underlying linear space of the reduces system can be taken to be $\cE$ which is determined by a decomposition 
$\linf=\cF\oplus\cG\oplus \cE$ with $\cD\cF=\cG$.

We now specialize a notion of equivalent reduction to the case where the underlying $Q$-manifolds are symplectic and $L_Q \omega=0$. Let $i:(N,\tilde Q)  \hookrightarrow (M,Q)$ be an equivalent reduction and $\tilde \omega= i^*\omega$ be nondegenerate.  
Pulling back \eqref{c45} to $N$ gives
\begin{equation}\label{formula75}
	i_{\tilde Q}\tilde{\omega}+d\tilde{S}_{BV}=0,
\end{equation}
where $\tilde{S}_{BV}= i^* S_{BV}$. Because we have assumed $\tilde \omega$
is invertible we arrive at a symplectic $Q$-manifold and hence at a BV system. It gives a toy model example of the  equivalent Lagrangian BV systems with the BV master action $\tilde S_{BV}$. In the case of local field theory the appropriate generalisation of such an equivalence was proposed in~\cite{Dresse:1990dj} and is known under the name of elimination of generalised auxiliary fields. 

Following~\cite{Dresse:1990dj} (see also~\cite{Barnich:2004cr}) let us spell out explicitly how this reduction can be phrased in terms of the action.
We restrict to local analysis and work in Darboux coordinates $\Phi,\Phi^*$. Suppose that by coordinate transformations we can introduce new Darboux coordinates $\phi^i$, $y^a$, $\phi^*_i$, $y^*_a$  such that
\begin{equation}\label{fo30}
	\left.\frac{\delta S_{BV}}{\delta y^a}\right|_{y^*_a=0}=0\Leftrightarrow y^a=Y^a[\phi,\phi^*],
\end{equation}
then $y^a$, $y^*_a$ are generalized auxiliary fields. Now consider a submanifold $N$ defined by
\begin{equation}\label{fo31}
	N:=\left\{\frac{\delta S_{BV}}{\delta y^a}=0,y^*_a=0\right\}\,.
\end{equation}
Note that the conditions can be written as $Qy^*_a=0$ and $y^*_a=0$ and hence we are dealing with an equivalent reduction and, moreover, by construction the induced 2-form $\tilde\omega$ is symplectic.  This gives us a new BV system, where $N$ is the field-antifield space. As coordinates on $N$ one can take $\phi^i$, $\phi^*_i$ and the BV action is
\begin{equation}
	S^{red}[\phi,\phi^*]=\left.S[\Phi,\Phi^*]\right|_N.
\end{equation}
By construction $S^{red}$ satisfies the master equation
\begin{equation}
	(S^{red},S^{red})^{red}=0.
\end{equation}
with respect to antibracket determined by $\tilde\omega$. Note that the antinbracket on $N$ can be also seen as a Dirac bracket if one considers $y^*_a$ and $Qy^*_a$ as the second-class constraints.

Finally, if $(M_\linf,Q)$ is a formal pointed $Q$-manifold associated to $(\linf,l_i)$ the above equivalent reduction become well-defined globally and give an alternative way to derive equivalent reductions, and minimal models in particular, for cyclic $L_\infty$-algebras. 

If $\linf$ splits as $\linf=\cG \oplus \cE \oplus \cF$,
where $\cE\subset \ker{\cD}$, $\cG\subset \im{\cD}$ and $\cD$
defines an isomorphic between $\cF$ and $\cG$ the restriction of cyclic structure to $\cE$ is nondegenerate or, in the $Q$-manifold language, the associated equivalent reduction results in symplectic $Q$-manifold. To see this
let us first observe that $\cG$ is orthogonal to $\cE$ and to itself. Moreover, by adjusting a basis in $\cE$ by vectors from $\cG$ one can always achieve that $\cF$ is orthogonal to $\cE$ and hence the restriction of cyclic structure to $\cE$ must be nondegenerate. The above argument works for $\linf$ finite dimensional or if extra conditions are imposed, e.g. if $\linf$ is Hilbert space.~\footnote{A simple counterexample is to take differential forms on $\fR^n$ with compact support and the inner product $\int \alpha \wedge \beta$.} To conclude, under some assumptions an equivalent reduction of symplectic $Q$ manifold gives again a symplectic $Q$-manifold. However, in the case of local field theory this is usually not the case as the cyclic structure is degenerate in general. 

\begin{example}
\normalfont Let us illustrate the reduction to a minimal model in terms of $Q$-manifolds by the example of the scalar field theory with cubic interaction. Since there are no gauge symmetries, the BV-action coincides with the classical action and is given by:
\begin{equation}\label{formula76}
	S_{BV}[\Phi,\Phi^*]:=\int_{\mathbb{R}^{1,3}}d^4x\left(\frac{1}{2}\Phi \mathrm{K}[\Phi]+\frac{\lambda}{3!}\Phi^3\right)\,,
\end{equation}
where $\mathrm{K}$ is the Klein-Gordon operator (generally speaking regularized) and $gh(\Phi)=0$, $gh(\Phi^*)=-1$. And we have a symplectic form
\begin{equation}\label{formula77}
	\omega=\int d^4x\delta{\Phi(x)}\delta\Phi^*(x).
\end{equation}
We choose the field configuration space as follows
\begin{equation}
    	\Phi=\phi+w,
\end{equation}
where $\phi\in \ker \mathrm{K}$ is a zero mode of $\mathrm{K}$, i.e. a solution of the Klein-Gordon equation and $w\in S(\mathbb{R}^{1,3})$ is a function whose derivatives rapidly decrease at infinity. The dual space of the antifields is decomposed accordingly:
\begin{equation}
    \Phi^*=\phi^*+w^*\,.
\end{equation}

Note that the chosen functional space is not closed under multiplication and we have to redefine the multiplication in such a way that the cyclicity condition is not violated. We do not dwell on the details of this procedure and simply assume that the multiplication is well-defined. A detailed discussion of subtleties can be found in ~\cite{Macrelli:2019afx,Arvanitakis:2019ald}. 

BV-action in terms of $\phi,w$ takes the form
\begin{equation}\label{formula80}
	S_{BV}[\phi,w]=\int_{\mathbb{R}^{1,3}}d^4x\left(\frac{1}{2}w \mathrm{K}[w]+\frac{\lambda}{3!}\phi^3+\frac{\lambda}{3!}w^3+\frac{\lambda}{2}w^2\phi+\frac{\lambda}{2}\phi^2w\right),
\end{equation}
and the equations of motion read as
\begin{equation}
	\mathrm{K}[w]+\frac{\lambda}{2}\phi^2+\lambda\phi w+\lambda w^2=0.
\end{equation}
Since $\mathrm{K}:S(\mathbb{R}^{1,3})\rightarrow S(\mathbb{R}^{1,3})$ is an invertible operator on the Schwartz space these equations can be solved perturbatively with respect to $w$. Decomposing $w$ in a series according to perturbation theory as $w=\sum^{\infty}_{k=1}\lambda^k w^{(k)}$ and considering second order in $\lambda$ we get
\begin{equation}\label{SCF1}
         	w^{(1)}=-\frac{1}{2!}\mathrm{K}^{-1}[\phi^2],\qquad w^{(2)}=\frac{1}{2!}\mathrm{K}^{-1}\left[\phi \mathrm{K}^{-1}[\phi^2]\right],
\end{equation}
and $w=\lambda w^{(1)}+\lambda^2 w^{(2)}+ o(\lambda^3)$. We see that $w$ and $w^*$ are the generalized auxiliary fields. Now we can define a submanifold $N:=\{w-\mathcal{V}(\phi)=0,w^*=0\}$, where $\mathcal{V}(\phi)$ is given by \eqref{SCF1}, and \textit{eliminate} $w$, $w^*$ (or what is the same a restriction to the submanifold $N$). Then BV-action takes the form
\begin{equation}\label{formula85}
	\tilde{S}_{BV}[\phi,\phi^*]=\frac{\lambda}{3!}\int_{\mathbb{R}^{1,3}}d^4x\phi^3-\frac{\lambda^2}{8}\int_{\mathbb{R}^{1,3}}d^4x\left(\phi^2 \mathrm{K}^{-1}[\phi^2]\right)+o(\lambda^3).
\end{equation}
By construction, the above reduced BV action is the initial action evaluated on the perturbative solutions of the free equations of motion and hence is the generating functional for the tree-level S-matrix in agreement with~\cite{Arefeva:1974jv,Jevicki:1987ax}.
\end{example}

\section{$L_\infty$-structures in local gauge theories}
\label{sec:loc}
\subsection{Jet-bundle BV description}

A systematic framework to describe local gauge theories such that the locality is realized in a manifest way is provided by the jet-bundle BV approach.  

Let us briefly recall the basics of jet bundles. Given a smooth locally-trivial fiber bundle $\cF\to X$, it defines its associated jet bundle $J^{k}(\cF)\rightarrow X$ for any integer $k\geq 0$,  which is also a bundle over $J^{k-1}(\cF)\rightarrow X$ (note that $J^0(\cF)=\cF$). It is useful to work with an infinite jet bundle limit $J^{\infty}(\cF)$ which is the projective limit of $J^{k}(\cF)$ as $k\to\infty$. As functions on $J^{\infty}(\cF)$ one takes local functions, i.e. those that arise from finite jet bundle (i.e. depend on finite number of coordinates). By some abuse of notations we denote local functions by $\fu(J^\infty(\cF))$). Any section $\alpha:X\to \cF$ has a canonical jet prolongation $\bar\alpha: X \to J^{\infty}(\cF)$. There is the canonical Cartan distribution $\cC \subset T(J^{\infty}(\cF))$ generated by tangent spaces to prolonged sections. In local coordinates, it is generated  by vector fields, called total derivatives,
\begin{equation}\label{j1}
    D_{\mu}=\frac{\partial}{\partial x^{\mu}}+\psi^A_{\mu}\frac{\partial}{\partial \psi^A}+\psi^A_{\mu\nu}\frac{\partial}{\partial \psi^A_{\nu}}+\dots,
\end{equation}
where $x^\mu,\psi^A_{\mu\ldots}$ are standard local coordinates on $J^{\infty}(\cF)$ induced by local coordinates $x^\mu$ on $X$ and local coordinates $\psi^A$ on the fibers of $\cF$.  If $\bar\alpha$ is a prolongation of $\alpha:X\to \cF$ then for any local function $f$ one has $(\bar\alpha^*(D_\mu f))=\dl{x^\mu}(\bar\alpha^*(f))$. Evolutionary vector fields are those preserving Cartan distribution. For instance, a vertical evolutionary $V$ satisfies $[D_{\mu},V]=0$.  The algebra of local differential forms on $J^{\infty}(\cF)$ is bigraded by horizontal and vertical form degree. Accordingly, de Rham differential decomposes as $d=\dh+\dv$. In local coordinates $\dh=dx^\mu D_\mu$. Further details and original references can be found in e.g.~\cite{Anderson1991,Krasil'shchik:2010ij}.

In the jet-bundle BV approach $\cF$ is a $\fZ$-graded fiber bundle over space-time manifold $X$, $\dim{X}=n$, whose sections are fields, ghosts (for ghosts), and antifields. The information about the equations of motion and (gauge for) gauge symmetries is encoded in the vertical evolutionary vector field $s$, called BRST differential, defined on $J^{\infty}(\cF)$ and satisfying
\begin{equation}
    \frac{1}{2}[s,s]=s^2=0\,, \qquad \gh{s}=1\,, \qquad [D_{\mu},s]=0\,.
\end{equation}
In what follows we refer to the above data as to a local BV system at the level of equations of motion. In the case of Lagrangian BV systems, in addition to $s$,
bundle $J^{\infty}(\cF)$ is equipped with a closed $(n,2)$-form $\omega$ of ghost degree $-1$ which is nondegenerate on vertical evolutionary vector fields. The systematic formulation  of Lagrangian BV in terms of jet bundles was put forward in~\cite{Barnich:1995ap,Barnich:2000zw}. The equation  of motion version was proposed in~\cite{Barnich:2004cr,Kaparulin:2011xy}

The data of a local BV system encode the equations of motion and the gauge (for gauge) symmetries of the underlying gauge system. More specifically, a section $\sigma:X\rightarrow \cF$ is called a solution if its jet prolongation $\bar{\sigma}:X\rightarrow J^{\infty}(\cF)$ satisfies
\begin{equation}\label{js}
    \bar{\sigma}^*\circ s=0,
\end{equation}
where $\bar{\sigma}^*$ denotes the pullback map determined by $\bar{\sigma}$. In other words solutions belong to the zero locus of $s$. Note that $\bar{\sigma}^*$ preserves the ghost degree and hence for a function $f$ of nonvanishing ghost degree one gets $(\bar{\sigma}^*)(f)=0$, simply because on $X$ there are no functions of nonvanishing degree. 

To describe gauge transformations of a given section $\sigma$ let $\epsilon$ be a vertical vector field on $\cF$ with  $\gh{\epsilon}=-1$ and $\bar\epsilon$ be its prolongation to $J^\infty(\cF)$.
It follows, ${\epsilon_\sigma}^*: C^{\infty}(J^{\infty}(\cF)) \to C^{\infty}(X)$ defined by ${\epsilon_\sigma}^*=\sigma^*\circ \bar\epsilon$
satisfies $\bar{\epsilon}^*(fg):= {\epsilon}^*(f){\sigma}^*(g)+(-1)^{|f|}{\sigma}^*(f)\bar{\epsilon}^*(g)$ for all $f,g\in C^{\infty}(J^{\infty}(\cF))$ and hence defines an infinitesimal gauge transformations of $\bar\sigma^*$:
\begin{equation}
\delta{\bar\sigma}^*=\bar{\epsilon}_\sigma^{*}\circ s\,.
\end{equation}
It is easy to check that $\bar\sigma^*+\bar{\epsilon}_\sigma^{*}\circ s$ is again a prolongation of a solution, indeed $\bar{\epsilon}_\sigma^{*}(s D_\mu f))=\bar\sigma^{*}(\bar\epsilon(s D_\mu f))=\dl{x^\mu} \bar{\epsilon}_\sigma^{*}(sf))$ because $\bar\epsilon$ is evolutionary and $\bar\sigma$ is a prolongation of $\sigma$. In a similar way one defines gauge for gauge symmetries. In the case of a Lagrangian BV system the above definition reproduces the equations of motion and gauge transformation encoded in the BV master action. The component form of the above construction of equations of motion and gauge transformations in terms of $s$ can be found in~\cite{Barnich:2004cr}, see also \cite{Lyakhovich:2004xd}.

\subsection{$L_\infty$ description of local BV systems}

We now turn to $L_\infty$-structure underlying a local BV system. Consider a vertical bundle $V\cF:=ker(\pi_*)\subset T\cF$, where $\pi_*$ is the differential of $\pi:\cF\to X$.
Given a fixed solution $\sigma_0:X\rightarrow\cF$ it determines a pullback bundle $\sigma^*_0 (V\cF)$ over $X$. The fiber of $\sigma^*_0 (V\cF)$ over $x\in X$ is a vertical subspace of $T_{\sigma_0(x)}E$. In this way we get a graded \textit{vector} bundle $\sigma^*_0 (V\cF)$ over $X$, which we denote as $\cN$ for brevity. The space of sections $\Gamma(\cN)$ is a module over $\cC^{\infty}(X)$.  Locally, it is given by functions on $X$ with values in a graded vector space, the typical fiber of $\cN$.

Let us restrict to local analysis and introduce new local coordinates on $\cF$ such that the chosen solution ${\sigma}_0$ is given by $\psi^{\prime A}=0$. In the associated jet-bundle coordinates $\psi^{\prime A}_{\mu\ldots}$ its prolongation $\bar\sigma_0$ is defined by $\psi^{\prime A}_{\mu\ldots}=0$. The expansion of the BRST-differential $s$ into the sum of terms homogeneous in $\psi^{\prime A}_{(\mu)}$ reads as:
\begin{equation}\label{j4}
    s=s_1+s_2+s_3+\dots,
\end{equation}
where $s_k$ is the term of total homogeneity degree $k-1$. For example the linear term $s_1$ is determined by its action on $\psi^A$ (here and below we omit the prime):
\begin{equation}\label{j5}
s_1(\psi^A)=\Omega^A_B(x)\psi^B+\Omega^{A|\mu}_B(x)\psi^B_{\mu}+\Omega^{A|\mu\nu}_B(x)\psi^B_{\mu\nu}+\ldots.
\end{equation}
This, in turn, is determined by a finite order linear differential operator: $\Omega:\Gamma(\cN) \to \Gamma(\cN)$ so that 
$s_1\Psi=\Omega \Psi$, where $\Psi=\psi^A e_A$ is a distinguished linear function on $J^\infty(\cF)$ with values in $\cN$ and where we assumed that $\Omega$ acts on fields by total derivatives. Let us give a component expression of the action of $\Omega$ on a given section $v=e_A v^A(x)$:
\begin{equation}
\Omega(v)=\Omega(e_A v^A)=e_B \Omega^B_A(x,\partial_x)v^A(x),
\end{equation}
where $e_A$ denotes a local frame of $\cN$ induced by a coordinate system $\psi^A$. Note that nilpotency of $s_1$ implies $\Omega^2=0$ so that $\Omega$ can be understood as a BRST-operator for some first-quantized system, see e.g.~\cite{Barnich:2004cr}.

In a similar way, $s_2$ is determined by 
\begin{equation}\label{j6}
     s_2(\psi^A)=  \frac{1}{2}\left(I^A_{BC}(x)\psi^B\psi^C+2I^{A|\mu}_{BC}(x)\psi^B_{\mu}\psi^C+I^{A|\mu\nu}_{BC}(x)\psi^B_{\mu}\psi^C_{\nu}+\dots\right),
\end{equation}
where dots denote terms quadratic in $\psi^A_{\dots}$ and  terms involving higher derivatives.  Again, $s_2$ is one to one with a bilinear graded-symmetric differential operator $I_2:\Gamma(\cN)\tensor \Gamma(\cN) \rightarrow\Gamma(\cN)$. More precisely, if $\Psi^M$ is a collective notation for fiber coordinates $\psi^A_{\mu\ldots}$ we can write $s_2 \psi^A=(I_2)^A_{MN}\Psi^M\Psi^N$ with $(I_2)^A_{MN}=(-1)^{\p{M}\p{N}}(I_2)^A_{NM}$, defining $I_2$ via   $I_2(v,w)=e_A(I_2)^A_{MN} \bar v^*(\Phi^M) \bar w^*(\Phi^N)$, where $\bar v,\bar w$ are prolongations of sections $v,w:X\to \cN$ and $\Phi^M$ are fiber coordinates on $J^\infty(\cN)$ defined in the same way as $\Psi^M$. It is instructive to decompose the multiindex $N$ as $(A,(\mu))$ and to write $I_2$ explicitly as:
\begin{equation}
    I_2(v,w)=I_2(e_B v^B,e_C w^C)=e_A(I_2)^A_{(B(\mu)),(C(\nu))}(x) \d_{(\mu)}v^B(x)\d_{(\nu)}w^C(x)\,.
\end{equation}

Proceeding as above, we arrive at the collection of multilinear graded-symmetric differential operators:
\begin{equation}    \Omega:\Gamma(\cN)\rightarrow\Gamma(\cN)\,, \qquad  I_{n\geq 2}:\Gamma^{\otimes n}(\cN)\rightarrow\Gamma(\cN)\,.
\end{equation}
Strictly speaking they are defined only locally, i.e. for sections on the coordinate patches of $X$. This gives an $L_{\infty}$-structure on $\Gamma(\cN)$ over a coordinate patch. On the overlap of patches  the respective $L_\infty$ structures are related by an $L_\infty$-isomorphism.

In terms of $\cN$ the field configurations can be identified with degree zero sections of $\cN$ (if fermionic fields are present this identification requires some care) and the condition that $\varphi\in\Gamma^{(0)}(\cN)$ is a solution takes the standard form:
\begin{equation}
\Omega\varphi+\sum_{i\geq 2}\frac{1}{i!}I_i(\varphi,\ldots,\varphi)=0\,.
\end{equation}
Similarly, a gauge transformation of $\varphi \in \Gamma^{(-1)}(\cN)$ can be also written in the form~\eqref{fo6} as
\begin{equation}\label{j8}
    \delta_{\epsilon}\varphi=\Omega\, \epsilon+\sum_{i\geq 2}\frac{1}{(i-1)!}I_i(\epsilon,\varphi,\dots,\varphi),
\end{equation}
where $\epsilon\in \Gamma^{(-1)}(\cN)$ is a gauge parameter. and $\varphi\in \Gamma^{(0)}(\cN)$ is a degree zero section. In a similar way we can define gauge for gauge transformations. Parameters of such transformations are elements of $\Gamma^{(-k)}(\cN)$, with $k\geq 2$.

In contrast to equations of motion and gauge symmetries the extension of the action functional~\eqref{formula9} to the local field theory setup is straightforward only at the formal level. The problem is that the cyclic structure (BV symplectic form) involves integration over $X$ and is well defined on sections with compact support only,  while interesting field configurations often do not belong to this class.  In the local field theory setup the standard way out is to work in terms of Lagrangians rather than action functionals. As we are going to see and discuss in more details in Section~\bref{sec:presymp}, the Lagrangians correspond to fiberwise (pre)symplectic structures. 

To summarise, the Taylor expansion of BRST-differential $s$ around some fixed solution gives us the $L_{\infty}$-structure on the space of sections $\Gamma(\cN)$ of the graded vector bundle $\cN$ induced by the chosen solution. All the structure maps are given by multilinear graded symmetric differential operators of ghost degree 1 and, in general, are defined only locally.

\subsection{Gauge PDEs}

A more flexible and invariant way to describe local gauge theories is based on the approach developed in~\cite{Barnich:2004cr,Barnich:2010sw,Grigoriev:2019ojp}. In this approach a general local gauge theory is represented in the  special 1st order form, generalizing the AKSZ formulation of topological models and the unfolded formalism of higher spin theories. More precisely, a gauge theory is encoded in the geometrical object known as gauge PDE \cite{Grigoriev:2019ojp}.

More specifically, a gauge PDE is a locally trivial fiber bundle $\pi:E \to T[1]X$ of $\fZ$-graded supermanifolds such that $\pi^*\circ \dx=Q\circ \pi^*$, where $\dx$ is the de Rham differential of $X$, seen as a vector field on $T[1]X$. In addition it is required to be locally equivalent to a non-negatively graded gauge PDE and to satisfy some regularity conditions. 

Let us spell out explicitly the structure of  $Q$ in general coordinates. Let $(x^{\mu}, \theta^{\mu}, \psi^A)$ be coordinates adapted to a local trivialization ($x^\mu$ originate from coordinates on $X$ and $\psi^A$ on the typical fiber). Using $\pi^*\circ \dx=Q\circ \pi^*$
one finds that $Q$ has the following form 
\begin{equation}\label{l5}
	Q=\theta^{\mu}\frac{\partial}{\partial x^{\mu}}+q^A(x,\theta,\psi)\frac{\partial}{\partial \psi^A}.
\end{equation}
In the case if the trivialisation can be chosen such that $q^A=q^A(\psi)$ the $Q$-bundle is locally-trivial and we are dealing with the globalization of an AKSZ sigma model (with possibly infinite-dimensional target space). See~\cite{Grigoriev:2019ojp,Barnich:2010sw} for further details. 

It is easy to see that gauge PDEs  encode gauge systems.  Indeed, field configurations are identified with sections $\sigma: T[1]X \to E$. Solutions are $Q$-sections, i.e. sections satisfying 
\begin{equation}\label{l2}
\dx\circ\sigma^*=\sigma^*\circ Q,
\end{equation}
The gauge transformation of a given section $\sigma$ is given by 
\begin{equation}\label{l3}
\delta\sigma^*=\dx\circ\epsilon^*+\epsilon^*\circ Q\,,
\end{equation}
where $\epsilon_\sigma^*:C^{\infty}(E)\rightarrow \cC^{\infty}(T[1]X)$ has degree $-1$, satisfies  $\epsilon_\sigma^*(fg):= \epsilon_\sigma^*(f)\sigma^*(g)+(-1)^{|f|}\sigma^*(f)\epsilon_\sigma^*(g)$ and 
$\epsilon_\sigma^*(\pi^* \alpha)=0$ for all $\alpha \in \fu(T[1]X)$. The latter condition exclude gauge transformations involving $X$ reparameterisations because such transformations generally transform sections to generic maps $T[1]X\to E$. Nevertheless they preserve equations \eqref{l2} and can be regarded as gauge symmetries of the parameterised version of the theory, see~\cite{Barnich:2010sw}. Note that it is sufficient to consider gauge parameters of the following form $\epsilon_\sigma^*=\sigma^*\circ Y$, where $Y$, $\gh{Y}=-1$ is a vertical vector field on $E$, in which case all the conditions are satisfied automatically.

\subsection{$L_\infty$ description of gauge PDEs}
\label{sec:gPDE-Li}

We now describe the general construction of $L_\infty$-structure associated to generic gauge PDEs.  Suppose we are given with a generic gauge PDE $(E,T[1]X,Q)$ and let $\sigma_0:T[1]X\rightarrow E$ be a fixed solution. Just like in the analogous discussion of local BV systems $\sigma_0$ defines a vector bundle $\cN=\sigma_0^*(VE)$ but now it is a graded vector bundle over $T[1]X$. Recall that $VE:=ker(\pi_*)\subset TE$, where $\pi:E\to T[1]X$ is a canonical projection.  The typical fiber of $\cN$ is just the typical tangent space to a fiber of $E$. Note that $\Gamma(\cN)$ is naturally a module over $\fu(T[1]X)$ which, in turn, can be seen as an exterior algebra of $X$.

Just like in the analysis of local BV systems it is convenient to introduce  new  coordinates  $(x^\mu,\theta^\mu,\psi^A)$ such that $\sigma_0^*(\psi^A)=0$. Expanding $Q$ in $\psi^A$ gives:
 \begin{equation}\label{l6}
 	Q=\theta^{\mu}\frac{\partial}{\partial x^{\mu}}+\sum_{i\geq 1}\frac{1}{i!}\psi^{B_1}\cdots\psi^{B_i}q^A_{B_1\cdots B_i}(x,\theta)\frac{\partial}{\partial \psi^A}\,,
 \end{equation}
 where coefficients $q^A_{B_1\cdots B_i}(x,\theta)$ are graded-symmetric in lower indexes.  A more detailed expression for the linear term in $Q$ reads as
 \begin{equation}\label{l17}
 	Q_1
=\psi^B\Omega^A_B(x)\dl{\psi^A}+\theta^{\mu}\left(\dl{x^{\mu}}+(-1)^{|\psi^B|}\psi^Bq^A_{\mu B}(x)\dl{\psi^A}\right)+\cdots,
 	\end{equation}
 where dots denote terms of higher orders in $\theta$. It is clear that the second term
 determines a linear connection. The nilpotency of $Q_1$ then implies
\begin{equation}
     \Omega^A_B(x)\Omega^B_C(x)=0.
 \end{equation}
 It also follows from $Q^2_1=0$  that the linear connection determined by $Q_1$ is flat modulo $\Omega$-exact terms. In particular, it induces a flat connection in the cohomology of $\Omega$ in $\Gamma(\cN)$.

Let us now focus on the space of sections $\Gamma(\cN)$ of $\cN$, which is a \textit{graded module} over $C^{\infty}(T[1]X)$ (recall that $\cN$ is a vector bundle over $T[1]X$). If we restrict ourselves to local analysis, as we do in what follows, $\Gamma(\cN)$ is isomorphic to differential forms on $X$ with values in the graded vector space $\cF$, which is a typical fiber of $\cN$. As a degree in $\Gamma(\cN)$ we take the sum of the form degree and the ghost degree in $\cF$, which, in turn, is induced from the degree in the fibers of $E$.

The homological vector field $Q$ induces a linear map $l_1: \Gamma(\cN)\rightarrow \Gamma(\cN)$. To define it in coordinate-free terms let us identify $\Gamma(\cN)$ as the quotient of the vertical vector fields on $E$, which we denote by $\mathfrak{X}_V(E)$, modulo those that vanish on $\sigma_0(T[1]X)$ seen as a submanifold of $E$.  This quotient is clearly a module over $\cC^{\infty}(T[1]X)$.  Then $l_1$ can be defined as follows:
\footnote{In the case where $\dim{X}=0$ this representation for $l_1$ was in \cite{Grigoriev:2000zg}.}
\begin{equation}
\label{l1inv}
l_1(\alpha)=[(\commut{Q}{v})]\,, \quad \alpha=[v]\,,
\end{equation}
where $[v]$ denotes an equivalence class of $v \in \mathfrak{X}_V(E)$. It is easy to see that $l_1$ is well-defined because $\commut{Q}{v}$ vanishes on $\sigma_0$ if $v$ itself  vanishing on $\sigma_0$. Indeed, $\sigma_0^*(Qvf)-(-1)^{\p{v}}\sigma_0^*(vQf)=0$ for all $f \in \cC^{\infty}(E)$ because $\sigma^*\circ v=0$ and $\sigma_0^* \circ Q= \dx \circ \sigma_0^*$. Introducing the local frame $e_A$ of $\cN$ the explicit expression for $l_1$ reads as:
\begin{equation}
l_1 \alpha = (\dx \alpha^A) e_A-(-1)^{|\alpha|}
\alpha^A(x,\theta)q^{C}_{A} e_C\,, \qquad  \alpha=  \alpha^A(x,\theta)e_A\,,
\end{equation}
and can be interpreted as a $\fZ$-graded version of  Quillen supeconnection in $\cN$, seen as a bundle over $X$. Superconnections were introduced in \cite{Quillen:1985vya} and turned out to be useful in the QFT context, see e.g.~\cite{Neeman:1990tfh,Roepstorff:1998vh}.   The linear in $\theta^\mu$ piece of $l_1$ reads as
\begin{equation}
    \nabla  \alpha =(\dx \alpha^A) e_A-(-1)^{|e_A|}\theta^{\mu}q^{C}_{\mu A}\alpha^A e_C
\end{equation}
and can be interpreted as a usual linear connection in $\cN$, seen as a vector bundle over $X$. 

In contrast to $l_1$ the higher $L_\infty$ maps on $\Gamma(\cN)$ depend on the choice of trivialization and are generally defined only locally. Because $\dx$ term in $Q=\dx+q$  contributes to $l_1$ only, the higher maps are determined by $q$ and are multilinear over $\fu(T[1]X)$.  To define them it is useful to employ the language of \textit{higher derived brackets} \cite{Voronov:2005,Voronov:2004tom,Batalin:1999gf}, see also \cite{KosmannSchwarzbach1996FromPA,Kosmann-Schwarzbach:2003skv}.

Suppose we are given a graded Lie algebra $V$ and its abelian subalgebra $U\subset V$. Let $P$ be a projector onto $U$ satisfying
\begin{equation}\label{hb1}
    P[v,w]=P[Pv,w]+P[v,Pw],
\end{equation}
for all $v,w\in V$. The identity \eqref{hb1} means that the kernel of $P$ is also a subalgebra in $V$. Let $\Delta \in V$  be an element satisfying $\frac{1}{2}[\Delta,\Delta]\equiv \Delta^2=0$ and $\Delta\in \ker{P}$. Then $U$ is naturally an $L_{\infty}$-algebra, whose $n$-ary map is given by the \textit{$n$-th derived bracket formula} \footnote{The formula \eqref{hb2} differs from the one in \cite{Voronov:2005} by prefactor $(-1)^{n+1}$ because we want to be in agreement with the notations of $L_{\infty}$-algebra, introduced above.}:
\begin{equation}\label{hb2}    l_n(v_1,\dots,v_n)=(-1)^{n+1}P[\dots[[\Delta,v_1],v_2],\dots,v_n],
\end{equation}
where $v_i\in U$ and $n=0,1,2,\dots$. Notice that \eqref{hb2} are always graded-symmetric. The generalized Jacobi identities are given by
\begin{equation}\label{hb3}    J_n(v_1,\dots,v_n)=P[\dots[[\Delta^2,v_1],v_2],\dots,v_n],
\end{equation}
where $J_n(v_1,\dots,v_n)$ denotes the $n$-th generalized Jacobi identity for $l_n$ so that for $\Delta^2=0$ these identities are satisfied. 

Let us specialize the above construction to the case at hand. As $V$ we take $\mathfrak{X}_V(E)$, the algebra of vertical vector fields on $E$.  To define $U$ we need to fix a torsion-free flat affine connection in $TE$ which preserves $V$ in the sense that $\bar\nabla_X Y\in V$ for any vector field $X$ and any vertical vector field $Y$. Such a connection always exists locally, for instance, given a local coordinate system $x^\mu,\theta^\mu,\psi^A$ one can simply take $\bar\nabla$ such that  
\begin{equation}
     \bar\nabla_{\frac{\small{\partial}}{\small{\partial\psi^A}}}\dl{\psi^B}=0  \,,\quad   \bar\nabla_{\frac{\small{\partial}}  {\small{\partial x^{\mu}}}} \dl{\psi^A}=0\,, \quad  \bar\nabla_{\frac{\small{\partial}}{\small{\partial\theta^{\mu}}}}\dl{\psi^A}=0\,.
\end{equation}
Then $U$ is defined to be a subalgebra of covariantly constant vertical vector fields, i.e. vector fields satisfying $\bar\nabla v=0,vx^\mu=0, v\theta^\mu=0$. It is easy to see that $U$ is isomorphic to $\Gamma(\cN)$ as any vertical vector field defined on $\sigma_0(X)$ has a unique covariantly constant extension to a neighbourhood of $\sigma_0(X)$. The projection $P$ then takes a vector field $v$ to the unique covariantly constant representative of the same equivalence class $[v]$.

It is easy to check that $l_1$ defined in this way indeed coincides with~\eqref{l1inv}. Let us spell out explicitly the expression for $l_2$:
\begin{equation}\label{HDB2}
l_2(v_\alpha,v_\beta)=P\commut{\commut{Q}{v_\alpha}}{v_\beta}\,,
\end{equation}
where $v_\alpha,v_\beta$ are the vector fields from $U$ representing sections $\alpha,\beta \in \Gamma(\cN)$.
The component expression for $l_2$ reads as
\begin{equation}
    l_2( \alpha,\beta)=(-1)^{|\alpha|+|\beta|+|e_A||\beta|+1} \alpha^A\beta^Bq^C_{AB}(x,\theta)e_C  \,,\quad \alpha=\alpha^A(x,\theta)e_A\quad \beta=\beta^A(x,\theta)e_A \, ,
\end{equation}

To summarize, a gauge PDE determines an $L_{\infty}$-algebra structure on the space $\Gamma(\cN)$. However, the $l_i$ with $i>1$ depend on the choice of trivialization and are defined only locally. On the overlap of coordinate patches they should be related by $L_{\infty}$-isomorphisms. Note that all the maps save for $l_1$ are multilinear over $\fu(T[1]X)$ while $l_1$ is a sum of $\dx$ and a linear piece. In other words, an $L_\infty$-algebra underlying a gauge PDE is almost algebraic in the sense that higher maps do not involve space-time derivatives.

A clear disadvantage of the above construction is that $L_\infty$ structure generally depends on the choice of trivialization and hence is in general defined only locally. It should be possible to develop a globally well-defined version by allowing $\nabla$ to be curved. This should extend the discussion of Appendix~\bref{app:Fedosov} to the case of $Q$-bundles. In this context it is also worth mentioning the related developments of \cite{Mehta_2015,Seol:2021tol, Arvanitakis:2022npu,Chatzistavrakidis:2023otk}.

In applications, it often happens that all the $q^A_{B\ldots}(x,\theta)$ are $\theta$-independent save for 
$q^{A}_{B}(x,\theta)=\Omega^A_B(x)+\theta^\mu q^A_{\mu B}(x)$. In this case $l_1$ takes the following form:
\begin{equation}
\label{nabla+omega}
    l_1=\nabla+\Omega\,,
\end{equation}
with $\nabla$ the flat connection in $\cN$ defined above and $\Omega$ a linear operator on $\Gamma(\cN)$ defines as $\Omega e_A=-(-1)^{\p{e_A}}\Omega_A^B e_B$. In this case the explicit form of the equations of motion and gauge transformations reads as
\begin{equation}
(\nabla+\Omega)\varphi+\sum_{i\geq2}\frac{1}{i!}l_i(\varphi,\ldots,\varphi)=0,\qquad \delta_{\epsilon}\varphi=(\nabla+\Omega)\epsilon+\sum_{i\geq2}\frac{1}{(i-1)!}l_i(\epsilon,\ldots,\varphi)\,,
\end{equation}
where $\varphi\in \Gamma^{(0)}(\cN)$, $\epsilon\in \Gamma^{(-1)}(\cN)$. In a similar way we
can define gauge for gauge transformations with parameters that are elements of $\Gamma^{(-k)}(\cN), k\geq 2$. Note that a rather concise description of gauge fields of arbitrary spin in constant curvature spaces as well as generic conformal gauge fields, see e.g.~\cite{Barnich:2004cr,Barnich:2006pc,Alkalaev:2008gi,Bekaert:2009fg,Alkalaev:2011zv,Bekaert:2013zya,Alkalaev:2018bqe} is based on $l_1$ of the form~\eqref{nabla+omega}.

Now we illustrate the flexibility of the formalism with more examples.  In the first example we consider the AKSZ formulation of Chern-Simons theory, which is viewed as a gauge PDE and show that the associated $L_\infty$-description is indeed the one recalled in the Example~\bref{ex:CS}.

\begin{example} \normalfont Let us consider a trivial $Q$-bundle $(E,Q,T[1]X)=(T[1]X,\dx)\times(\algg[1],q)$, where $(\algg[1],q)$ is CE complex of a Lie algebra $\algg$, considered as a $Q$-manifold. If $c^{a}, gh(c^{a})=1$ are coordinates on $\algg[1]$, homological vector field $q$ is given by $q=-\frac{1}{2}U^{a}_{bc}c^{b}c^{c}\frac{\partial}{\partial c^{a}}$, where $U^{a}_{bc}$ are structure constants of $\algg$. Then $Q$-structure on $E$ takes the form
\begin{equation}\label{l8}
		Q=\theta^{\mu}\frac{\partial}{\partial x^{\mu}}-\frac{1}{2}U^{a}_{bc}c^{b}c^{c}\frac{\partial}{\partial c^{a}}\,.
\end{equation}
Now let us consider a fixed solution $\sigma_0:T[1]X\rightarrow E$. If $(x^{\mu},\theta^{\mu})$ are local coordinates on $T[1]X$ this map can be written as $\sigma^*_0(c^{a})=\theta^{\mu}A^{a}_{\mu}(x)$. Introducing new coordinates $c^{\prime a}$ in a fiber such that $\sigma^*_0(c^{\prime a})=0$ and omitting the prime the $Q$-structure takes the form:
\begin{equation}\label{l31}		Q=\theta^{\mu}\left(\frac{\partial}{\partial x^{\mu}}-q^{a}_{\mu b}(x)c^{b}\frac{\partial}{\partial c^{a}}\right)-\frac{1}{2}U^{a}_{bc}c^{b}c^{c}\frac{\partial}{\partial c^{a}},
	\end{equation}
where $q^{a}_{\mu b}(x)\equiv U^{a}_{cb}A^{c}_{\mu}(x)$ is interpreted as the coefficients of the flat connection.
so that the first term can be considered as a flat covariant derivative in the representation of $\algg[1]$ on functions in $c^a$.

Homological vector field~\eqref{l31} determines an analogue of $L_{\infty}$ - structure on the space of sections of the vector bundle $\cN$ which is $VE \subset TE$ pulled back to $T[1]X$ by $\sigma_0$. Representing $u,v\in\Gamma(\cN)$ as vertical vector fields $u^a(x,\theta)\dl{c^a}$
and $v^a(x,\theta)\dl{c^a}$ and using~\eqref{l1inv} and \ref{HDB2} 
we get the following explicit expressions for $l_1$ and $l_2$ 
\begin{equation}\label{l32}
		l_1 (v^a \dl{c^a})=(\dx v^a+ \theta^\mu q_{\mu b}^a v^b)\dl{c^a}\,, \qquad l_2(u,v)=(-1)^{\p{u}}U_{ab}^c u^a v^b \dl{c^c}
\end{equation}
Note that $l_1$ plays a role of a linear connection in $\cN$ (note the sign difference with the respective term in \eqref{l31}, which arises because the linear term in $Q$ defines the covariant differential in the bundle dual to $\cN$),  while $l_2$ gives a Lie bracket on sections (i.e. the  Lie bracket in $\algg$ extended to $\Gamma(\cN)$).  Identifying sections of $\cN$
as $\algg$-valued differential forms on $X$, one can rewrite~\eqref{l32} as:
\begin{equation}
l_1 v= \dx v+\commut{A}{v}\,,
\qquad 
l_2(u,v)=(-1)^{\p{u}} \commut{u}{v}\,,
\end{equation}
in agreement with Example~\ref{ex:CS}. 
Note that here we deal with a slightly more general $L_\infty$ structure as it allows for a nontrivial background solution. Analogous considerations apply to any gauge system of AKSZ-type.
\end{example}

The next example is the gauge PDE representation of a PDE that is defined in terms of its equation manifold equipped with the Cartan distribution~\cite{Vinogradov1981}, see~\cite{Krasil?shchik-Lychagin-Vinogradov,Krasil'shchik:2010ij} for a review.

\begin{example} \normalfont PDEs can be seen as gauge PDEs, where the ghost degree is horizontal, i.e. all fiber coordinates have vanishing degree~\cite{Grigoriev:2019ojp}. Let $(E,Q,T[1]X)$
such a gauge PDE and $(x^\mu,\theta^\mu,\psi^A)$ be local coordinates on $E$, where $\psi^A$ are fiber coordinates. In this case a general $Q$ has the form
\begin{equation}\label{l9}
	Q=\theta^{\mu}\left(\frac{\partial}{\partial x^{\mu}}-\Gamma^A_{\mu}(x,\psi)\frac{\partial}{\partial \psi^A}\right),
\end{equation}
where $\Gamma^A_{\mu}(x,\psi)$ can be interpreted as coefficients of the connection $1$-form and $Q^2=0$ implies that the connection is flat, i.e. the Cartan distribution is involutive.  We again take a fixed section $\sigma_0:T[1]X\rightarrow E$ and introduce a new coordinate system $(x,\theta,\psi^A)$ such that $\sigma^*_0(\psi^A)=0$. Expanding \eqref{l9} about $\psi^{A}=0$ we get
\begin{equation}\label{l16}
	Q=\theta^{\mu}\left(\frac{\partial}{\partial x^{\mu}}-q^A_{\mu B}(x)\psi^B\frac{\partial}{\partial\psi^A}\right)-\sum_{i\geq 2}\frac{1}{i!}\theta^{\mu}q^A_{\mu B_1\cdots B_i}(x)\psi^{B_1}\cdots\psi^{B_i}\frac{\partial}{\partial \psi^{A}}\,.
\end{equation}
Again, the first term defines a linear connection on the associated vector bundle $\cN\rightarrow T[1]X$ induced by the section $\sigma_0$.  More precisely, identifying sections of $\cN$ as vector fields of the form $v=v^A(x,\theta)$ and applying \eqref{l1inv} to the above $Q$ gives:
\begin{equation}
l_1v=(\dx v^A+(-1)^{\p{v}}\theta^\mu q^A_{\mu B} v^B)\dl{\psi^A}
\end{equation}
where for simplicity we assumed $\p{\psi^A}=0$. In particular, solutions of the linearized equation are the degree $0$ section $v=v^A(x)e_A$ that are covariantly constant:
\begin{equation}
l_1v= (\dx v^A+\theta^\mu q^A_{\mu B} v^B)\,,
\end{equation}
in agreement with the standard picture and the definition of solutions as $Q$-sections. The higher maps $l_{i\geq 2}$ correspond to the nonlinearities of the underlying PDE. 
\end{example}

The last example is the AKSZ description of the \textit{off-shell higher spin fields} \cite{Vasiliev:2005zu,Barnich:2004cr,Grigoriev:2006tt}
\begin{example}
\normalfont Let $T^*V$ be cotangent bundle over $V\equiv \fR^n$. Let us denote by $y^a,p_b$ standard coordinates on $T^*V$ so that the 
canonical symplectic structure reads as $dp^a \wedge dy^a$.
This defines a Weyl star product for functions on $T^*V$; for instance, $[y^a,p_b]_*=\hbar\delta^a_b$.

We then extend our space by Grassmann odd ghost variable $c$, $gh(c)=1$ and take $\algA$ to be the algebra of formal power series in $\hbar$ with coefficients in functions in $y,p,c$ that are polynomials in $p_a,c$ with coefficients in formal power series in $y^a$. In what follows we do not explicitly indicate dependence on $\hbar$.

Consider a graded manifold $M_{\algA}$ associated with $\algA[1]$. Coordinates on this manifold are coefficients of the generating function $\Psi(y,p,c)=\textbf{A}(y,p)+c\textbf{F}(y,p)$.
The ghost degree prescription is such that $\gh{\Psi}=1$ and hence $gh(\textbf{A})=1$, $gh(\textbf{F})=0$.
$M_{\algA}$ is naturally a $Q$-manifold with the $q$-structure defined via:
\begin{equation}
q\Psi=-\frac{1}{2\hbar}\qcommut{\Psi}{\Psi}\,.
\end{equation}
We then take the following gauge PDE: $E\to T[1]X$, with $(E,Q)=(M_{\algA},q)\times (T[1]X,\dx)$, where $\dim{X}=n$ (in this case it is an AKSZ sigma model). In the component terms $Q$ reads as
\begin{equation}\label{hs2}
    Q(x^{\mu})=\theta^{\mu},\quad Q(\textbf{A})=-\frac{1}{2\hbar}[\textbf{A},\textbf{A}]_*, \quad Q(\textbf{F})=-\frac{1}{\hbar}[\textbf{A},\textbf{F}]_*.
\end{equation}
If $\sigma$ is a section determined by $\sigma^*(\textbf{A})=\theta^{\mu}{A}_{\mu}(x|y,p)\equiv A(x,\theta|y,p)$  and  $\sigma^*(\textbf{F})={F}(x|y,p)\equiv F(x|y,p)$ the equations of motion read as
\begin{equation}
\label{hs-eom}
    \dx A+\frac{1}{2\hbar}[A,A]_*=0,\qquad 
  \dx F+\frac{1}{\hbar}[A,F]_*=0\,,
\end{equation}
The gauge transformations read as
\begin{equation}
    \delta A=\dx\lambda+\frac{1}{\hbar}[A,\lambda]_*,\qquad 
    \delta F=\frac{1}{\hbar}[F,\lambda]_*\,,
\end{equation}
where $\lambda=\lambda(x|y,p)$ is a gauge parameter. These equations are familiar in the context of Fedosov deformation quantization  of cotangent bundles~\cite{Bordemann:1997ep,Fedosov2001} and in the higher spin theory context~\cite{Vasiliev:2005zu}. The above AKSZ-like description originates from~\cite{Grigoriev:2006tt}  (see also~\cite{Barnich:2004cr}).

We now pick a particular solution and spell-out explicitly the emerging $L_\infty$ structure. Namely, we take 
\begin{equation}\label{hs7}
    A_0=\theta^{\mu}\left(e^a_{\mu}p_a+\omega^a_{\mu b}y^bp_a\right), \quad F_0=\frac{1}{2}\eta^{ab}p_ap_b,
\end{equation}
where $e, \omega, \eta$ are identified as, respectively,  the vielbein, the Lorentz connection, and the standard Minkowski metric. The equation of motion imply:
\begin{equation}\label{hs8}
    \dx e^a+\omega^a_b e^b=0, \qquad \dx\omega^a_b+\omega^a_c\omega^c_b=0, \qquad \omega^a_c\eta^{cb}+\eta^{ac}\omega^b_c=0\,,
\end{equation}
so that the underlying geometry is flat and is that of a Minkowski space. Introducing new fiber coordinates (such that the vacuum solution \eqref{hs7} reads $\textbf{A}=0, \textbf{F}=0$ and expanding $Q$ in the fiber coordinates 
\begin{equation}
\begin{gathered}
    Q(\textbf{A})=-\frac{1}{2\hbar}[\textbf{A},\textbf{A}]_*-\theta^{\mu}\left(e^a_{\mu}\frac{\partial}{\partial y^a}+\omega^a_{\mu b}y^b\frac{\partial}{\partial y^a}-\omega^a_{\mu b}p_a\frac{\partial}{\partial p_b}\right)\textbf{A},\\
    Q(\textbf{F})=-\frac{1}{\hbar}[\textbf{A},\textbf{F}]_*-p^a\frac{\partial}{\partial y^a}\textbf{A}-\theta^{\mu}\left(e^a_{\mu}\frac{\partial}{\partial y^a}+\omega^a_{\mu b}y^b\frac{\partial}{\partial y^a}-\omega^a_{\mu b}p_a\frac{\partial}{\partial p_b}\right)\textbf{F}\,, \qquad Q(x^{\mu})=\theta^{\mu}\,.
\end{gathered}
\end{equation}
Applying the general procedure we infer an underlying $L_\infty$-structure on the section of $\cN$.  In the case at sections of $\cN$ are differential forms on $X$ with values in $\algA[1]$ which, in turn, is identified with functions in $y,p,c$.  We find
\begin{equation}\label{hs11}
	l_1(\Phi)=\left(\nabla+\sigma-cp^a\frac{\partial}{\partial y^a}\right)\Phi,\qquad 
 \qquad l_2(\Phi,\Phi)=\frac{(-1)^{\p{\Phi}}}{2\hbar}\qcommut{\Phi}{\Phi}\,,
\end{equation}
where $\Phi=\Phi(x,\theta|y,p,c)$ and  
\begin{equation}\label{hs12}
     \nabla=\theta^{\mu}\left(\frac{\partial}{\partial x^{\mu}}-\omega^a_{\mu b}y^b\frac{\partial}{\partial y^a}+\omega^a_{\mu b}p_a\frac{\partial}{\partial p_b}\right),\quad \sigma=-\theta^{\mu}e^a_{\mu}\frac{\partial}{\partial y^a}\,.
\end{equation}
The pair $(\Gamma(\cN),l_1)$ can be interpreted as the first-quantized constrained system describing a scalar relativistic particle. In so doing $\Phi=\Phi(x,\theta|c,y,p)\in\Gamma(\cN)$ is interpreted as a wave function of such a system and $l_1$ as its BFV-BRST operator, see \cite{Grigoriev:2006tt} for more details.

System $(\Gamma(\cN),l_1,l_2)$ describes the non-linear deformation of the linear system $(\Gamma(\cN),l_1)$ describing the off-shell higher spin fields  on the Minkowski space-time.  In order to describe on-shell higher-spin fields in this way one should require  $\Phi(x,\theta|c,y,p)$ to be totally traceless. In the present framework this can be achieved by subjecting $\Phi$ to the following algebraic conditions
\begin{equation}\label{hs13}
    \frac{\partial}{\partial y^a}\frac{\partial}{\partial p_a}\Phi=\frac{\partial}{\partial y^a}\frac{\partial}{\partial y_a}\Phi=\frac{\partial}{\partial p_a}\frac{\partial}{\partial p^a}\Phi=0\,.
\end{equation}
Note that the bilinear map $l_2$ does not descent to the above subspace in a consistent way, signalling the well known problem of higher spin interactions.

In the above example we took $e,\omega$ of the flat Minkowski space. However, any invertible frame field $e$ and a torsion-free Lorentz connection $\omega$ can be lifted to a solution of~\eqref{hs-eom}, using the procedure analogous to that employed in Fedosov quantization, see e.g.~\cite{Grigoriev:2006tt,Basile:2022nou}. The respective linearized system $(\Gamma(\cN),l_1)$ then describes off-shell higher spin fields on the curved background. Detailed study and applications of (various modifications of) this system can be found in~\cite{Grigoriev:2012xg,Bekaert:2017bpy,Basile:2022nou}.
\end{example}

\subsection{Lagrangians and presymplectic structures}
\label{sec:presymp}

So far we have focused on the description of local gauge systems at the level of equations of motion. Now we would like to briefly discuss how Lagrangian local BV systems can be described in the gauge PDE approach. It turns out that the BV Lagrangian and the BV symplectic structure can be encoded in a compatible presymplectic structure defined on the gauge PDE. 

More specifically,  let $(E,Q,T[1]X)$ be a gauge PDE. Vertical differential forms on $E$ are equivalence classes of forms modulo the ideal $\cI$ in the exterior algebra, which is generated by elements of the form $\pi^*\alpha$,
where $\alpha$ is a form of positive form degree on $T[1]X$. Note that the notion is general and applies to generic fiber bundles as well. Note also that the de Rham differential and the Lie derivative along projectable vector fields are well defined on equivalence classes (vector field $Y$ on $E$ projects to $y$ on $T[1]X$ if $\pi^*\circ y=Y \circ \pi^*$).

By definition a compatible presymplectic structure on $(E,Q,T[1]X)$ is a closed vertical 2-form $\omega$ of ghost degree $\dim{X}-1$ satisfying  $d\omega=0$ and $L_Q \omega \in \cI$. We also fix a presymplectic potential $\chi$ such that $\omega=d\chi$ and assume that $i_Q L_Q \omega=0$. It follows there exists $\hL$ such that 
$i_Q\omega+d\hL\in \cI$ and we have all the ingredients to
define BV action as
   \begin{equation}
\label{aksz-action}
S_{BV}[\hat\sigma]=\int_{T[1]X} \hat\sigma^*(\chi)(\dx) +\hat\sigma^*(\hL) \,, 
\end{equation}
where $\hat\sigma$ is a super-section of $E \to T[1]X$. In this form the expression is identical to that of the AKSZ action but the present framework is more general. The classical action is obtained by replacing $\hat\sigma$ with
a section $\sigma$ or, equivalently, setting to zero all the nonvanishing ghost degree coordinates on the space of supersection. 

The interpretation of the above BV-AKSZ-like action is that one should take a symplectic quotient of the space of supermaps for the AKSZ-like  presymplectic structure, induced by $\omega$ on the space of super-sections, giving a genuine local BV system whose interpretation is standard. Of course one should assume that the induced presymplectic structure is regular.  For further details on the presymplectic gauge PDEs we refer to~\cite{Dneprov:2022jyn,Grigoriev:2022zlq} (see also~\cite{Alkalaev:2013hta,Grigoriev:2016wmk,Grigoriev:2020xec} for earlier relevant works).

As we have already seen, given a fixed solution  $\sigma_0$ to $(E,Q,T[1]X)$ the space of sections of the associated vector bundle $\cN=\sigma_0^*(VE)$ is equipped with $L_\infty$-structure (as before we restrict to local analysis). It is clear that $\omega$ defines a fiberwise symplectic form $\omega_0$ on $\cN$.
In fact it also defines a higher multilinear maps (higher jets of $\omega$). More precisely, working in coordinates  $x^\mu,\theta^\mu,\psi^A$ such that $\sigma^*_0(\psi^A)=0$ 
the presymplectic form can be written as
\begin{equation}\label{sm1}
    \omega=\frac{1}{2}d\psi^Ad\psi^B\omega_{AB}(x,\theta,\psi),
\end{equation}
and its Taylor expansion in fiber coordinates reads as:
\begin{equation}\label{sm3}    \omega=\omega_0+\omega_1+\omega_2+\dots=\sum_{i\geq 0}\omega_i,
\end{equation}
where a superscript denotes homogeneity in the fiber coordinates.  While $\omega_0$ is a bilinear (over $\fu(T[1]X)$) form on $\Gamma(\cN)$, $\omega_{k>0}$
can be interpreted as  higher multilinear maps $\Gamma^{\tensor (k+2)}(\cN)\to \Gamma(\cN)$. Note that, in general, these maps are  defined only locally and depend on the choice of trivialisation.

The compatibility conditions, $\omega_i$ \ruth{ought to satisfy,} are obtained by expanding $L_Q\omega\in \cI$.  In particular, for the linear operator $\bar\Omega=l_1-\dx$ determined by the linear part of $Q$ one gets:
\begin{equation}
\begin{gathered}
\omega_0(\bar\Omega V,W)+(-1)^{\p{V}}\omega_0(V, \bar\Omega W)=(\dx \omega_0)(V,W)\,,\\ (\dx \omega_0)(V,W)=
\dx (\omega_0(V,W))- \omega_0(\dx{V},W)
-(-1)^{\p{V}}\omega_0(V,\dx{W})\,.
\end{gathered}
\end{equation}
The first relation is naturally interpreted as the $\bar\Omega$-invariance of the cyclic structure modulo total derivatives.

Before giving a simple example of this structure let us note that the above interpretation of presymplectic gauge PDEs suggests that all the coordinates which are in the kernel of $\omega$ can be safely ignored as the respective fields never show up in the BV master action.  This suggests the notion of a weak presymplectic gauge PDE, where the condition $Q^2=0$ is replaced with $\omega(Q,Q)=0$, see~\cite{Grigoriev:2022zlq} for more details. In the $L_\infty$ language this corresponds to "pre-$L_\infty$" structure encoded by such $Q$ and possibly degenerate compatible bilinear forms $\omega_i$. Note that for such weak presymplectic gauge PDEs the solutions are defined as those of the Lagrangian system with action~\eqref{aksz-action}, i.e. $\sigma_0$ is a solution if $\sigma_0^*(\omega_{AB})(\dx \sigma_0^*(\psi^B)-\sigma_0^*(Q^B))=0$.

\begin{example}
\normalfont Consider as $\cN$ a trivial vector bundle over $T[1]X$, with $X$ a Minkowski space, with the fiber being $\left((X^*\wedge X^*)\oplus \fR[-1]\right)\tensor \algg$, where $\algg$ is a reductive Lie algebra equipped with the nondegenerate invariant inner product denoted by $\tr$. A generic section of this bundle can be written as:
\begin{equation}
\Phi=C(x,\theta)\tensor e_1+F_{\mu\nu}(x,\theta)\tensor e_0^{\mu\nu}
\end{equation}
where $e_1,\gh{e_1}={-1}$ and $e^{\mu\nu}_0, \gh{e^{\mu\nu}_0}=0$ denote basis in $\fR[-1]$ and $X^*\wedge X^*$ respectively,
and where by some abuse of notations $F_{\mu\nu},C$ are $\algg$-valued. 

Thanks to the linearity over $\theta^\mu$ it is enough to define pre-$L_\infty$ maps for sections of the form $\alpha(x) e_1$
and $\beta_{\mu\nu}(x) e_0^{\mu\nu}$:
\begin{equation}
\label{l-maps-YM}
\begin{gathered}
l_1(\alpha e_1)=(\dx\alpha)\, e_1\,,\qquad 
l_1(\beta_{\mu\nu} e^{\mu\nu}_0)=(\dx \beta_{\mu\nu})\,e^{\mu\nu}_0-\half \theta^\mu\theta^\nu \beta_{\mu\nu} e_1\\
l_2(\alpha e_1,\gamma e_1)=-\commut{\alpha}{\gamma} e_1\,,\qquad
l_2(\beta_{\mu\nu}e^{\mu\nu}_0,\alpha e_1)= 2\commut{\beta_{\mu\nu}}{\alpha} e_0^{\mu\nu}\,,\qquad
l_2(\beta_{\mu\nu}e^{\mu\nu}_0,\gamma_{\rho\sigma} e^{\rho\sigma}_0)=0\,.
\end{gathered}
\end{equation}
Here and below we drop the tensor product sign.
Note that $l_1,l_2$ does not satisfy the Leibniz rule. For instance for $v=B_{\mu\nu}(x)e^{\mu\nu}_0$ one gets  
$l_1(l_2(v,v)+l_2(l_1(v),v)+l_2(v,l_1(v))=\theta^{\mu}\theta^{\nu}[B_{\mu\nu},B_{\rho\sigma}]e^{\rho\sigma}_0$. As we are going to see the above $l_1,l_2$ are determined by a $Q$-structure whose square is not equal to zero. 

The symplectic structure is field-independent, i.e. $\omega=\omega_0$ and is given by
\begin{equation}
\label{omega-YM}
\omega(\alpha e_1, \beta e_1)=0\,, \qquad \omega(\alpha e_0^{\mu\nu}, \beta e_0^{\rho\sigma})=0 \,, \qquad \omega(\alpha e_0^{\mu\nu}, \beta e_1)=\frac{1}{2}\tr{(\alpha\beta)}\epsilon^{\mu\nu}{}_{\rho\sigma}\theta^\rho\theta^\sigma \,, 
\end{equation}
where $\alpha,\beta$ are generic $\algg$-valued functions in $x,\theta$ and for simplicity we restricted ourselves to $dim(X)=4$, extension to generic dimension $>1$ is straightforward. Now, the analog of an $L_\infty$-action functional can be written as
\begin{equation}\label{YM4}
    S[\Phi^{(0)}]=\int d^4xd^4\theta\left(\frac{1}{2}\omega(\Phi^{(0)},l_1 \Phi^{(0)})+\frac{1}{3!}\omega(\Phi^{(0)},l_2(\Phi^{(0)},\Phi^{(0)}))\right)\,.
\end{equation}
where $\Phi^{(0)}=A_\mu(x)\theta^\mu e_1+B_{\mu\nu}(x)e_0^{\mu\nu}$ is a component of $\Phi$ of total ghost degree $0$. In terms of component fields one gets (modulo total derivative):
\begin{equation}\label{YM5}
    S[A,B]=\int d^4x\tr\left(B^{\mu\nu}(\partial_{\mu}A_{\nu}-\partial_{\nu}A_{\mu}+\commut{A_{\mu}}{A_{\nu}}\right)-\frac{1}{2}B^{\mu\nu}B_{\mu\nu}),
\end{equation}
which is a standard first-order action of the YM theory, see e.g.~\cite{McKeon:1994ds}.

Let us for completeness give a description of the above system as a presymplectic gauge PDE. The underlying bundle $E$ is just $\cN$ where the fibers are now understood as graded manifolds rather than graded linear spaces. Denoting the fiber coordinates associated to $e_1$ and $e^{\mu\nu}_0$ by $C$ and $F_{\mu\nu}$ respectively, the pre-$Q$-structure is defined as
\begin{equation}\label{YM1}
\begin{gathered}     Q(x^{\mu})=\theta^{\mu},\quad Q(C)=-\frac{1}{2}[C,C]+\frac{1}{2}\theta^{\mu}\theta^{\nu}F_{\mu\nu},\qquad 
   Q(F_{\mu\nu})=[F_{\mu\nu},C]
\end{gathered}
\end{equation}
and the presymplectic structure is given by \cite{Alkalaev:2013hta}:
\begin{equation}\label{YM3}
\omega=d\chi\,, \qquad   \chi=\frac{1}{2}\tr(\epsilon_{\mu\nu\rho\sigma}\theta^\rho\theta^\sigma 
 F^{\mu\nu}\, dC)\,.    
\end{equation}
Note that $Q^2\neq 0$ but the weaker condition holds: $\omega(Q,Q)=0$, see~\cite{Grigoriev:2022zlq}. It is easy to check that with this choice of $E,X,Q,\chi$  action~\eqref{aksz-action} indeed gives first-order action \eqref{YM5}. Taking symplectic quotient of the space of super-sections results in the full-scale BV formulation of the YM theory, see~\cite{Grigoriev:2022zlq} for more details.

Finally, the $L_\infty$-description based on~\eqref{l-maps-YM},\eqref{omega-YM} is obtained by expanding around the  vacuum solution $\sigma_0$ defined by $\sigma_0^*(C)=0,\sigma_0^*(F_{\mu\nu})=0$.

\end{example}

\section*{Acknowledgements} 
We wish to thank Alex Arvanitakis, 
Athanasios Chatzistavrakidis, Larisa Jonke, Ivo Sachs,  Richard Szabo, Alexander Verbovetsky, Victor Lekeu, Theodore Voronov, and especially Georgy Sharygin. The work was supported by the Russian Science Foundation grant 18-72-10123. Part of this work was done when MG participated in the program ``Higher Structures, Gravity and Fields'', Mainz Institute for Theoretical Physics, Johannes Gutenberg University, Mainz, Germany.

\appendix
\section{Globalisation of the $L_\infty$ structure}
\label{app:Fedosov}

Here we briefly recall the Fedosov-like construction for cotangent bundles and show that it can be used to obtain a globally well-defined curved $L_\infty$-structure on $\Gamma(TM)$ induced by the homological vector field $Q$ on $M$.

Let $M$
be a graded supermanifold and let $\Gamma$ be a torsion-free affine connection in $TM$.
Consider the vector bundle $W(M)=\hat S(T^*M) \tensor TM$, whose fiber we identify with a Lie algebra of formal vector fields on the tangent space. Here $\hat S (T^*M)$ denotes a completion of the graded symmetric tensor algebra of $T^*M$.   There is a natural  Lie bracket  $\commut{-}{-}_W$ on $\Gamma(W)$ determined by the fiberwise Lie bracket of formal vector fields. To make contact with the  version~\cite{Bordemann:1997er,Fedosov2001} (see also~\cite{Grigoriev:2006tt,Grigoriev:2016bzl,Basile:2022nou} for applications in the field theory context) of Fedosov quantization~\cite{Fedosov:1994}
specialised to the case of  cotangent bundles, note the identification of $\Gamma(W)$ with a Lie subalgebra of $\hat S(T^*M) \tensor S(TM)$, equipped with the standard fiberwise Weyl star-product.

The standard fact is that there exists a flat nonlinear connection $D$ in $W$ of the following form
\begin{equation}
D=\nabla+\commut{r}{\cdot}_W\,, \qquad D(Dv)=0\,\quad \forall v\in\Gamma(W)
\end{equation}
where $\nabla$ is the covariant differential in $\Gamma(W)$ determined by $\Gamma$ extended to $W$ and $r$ is a $W$-valued $1$-form on $M$.  The existence of such $r$  is a simple corollary of the Fedosov-like construction for the cotangent bundles: in this case the minimal Fedosov-like connection takes values in $W$ seen as subbundle of the Weyl bundle (i.e. is linear in momenta).

 Furthermore, Any section $v_0$ of $TM$ has a unique extension to $v\in\Gamma(W)$ such that $Dv=0$ and $v|_{TM}=v_0$, where $v|_{TM}$ denotes a projection of $v$ to $\Gamma(TM)$ seen as a direct summand in $\Gamma(W)$. Moreover,
\begin{equation}
\left(\commut{v}{u}_{W}\right)\big|_{TM}=\commut{v_0}{u_0}\,,
\end{equation}
where $\commut{v_0}{u_0}$ is the standard Lie bracket on $\Gamma(TM)$. Note that $D\commut{v}{u}_W=0$ because by construction $Dv=0=Du$ and both $\nabla$ and $\commut{r}{-}_W$ differentiate $\commut{-}{-}_W$.

Let $v_0,u_0$ are two  vector fields on $M$  such that $\commut{u_0}{v_0}=0$. Their unique covariantly constant lifts $v,u$ satisfy $\commut{v}{u}_W=0$ because $D\commut{v}{u}_W=0$ and $\commut{v_0}{w_0}=0$ by the assumption. If $M$ is a $Q$-manifold, we apply this to $v_0=u_0=Q$ and find that $\commut{q}{q}_W=0$, where $q \in \Gamma(W)$ is a unique lift of $Q$ satisfying $Dq=0$.  At the same time, this is equivalent to having a formal homological vector field on each fiber of $TM$, giving a curved-$L_\infty$ structure on $\Gamma(TM)$.

\section{Notations and conventions}

Consider a $\mathbb{Z}$-graded vector space $\linf:=\bigoplus_{k\in\mathbb{Z}}\linf_k$. Elements of subspace $\linf_k$ are called homogeneous of degree $k$, and we shall denote the Grassmann parity by $|\phi|$, where $\phi\in \linf$. The degree shift by $k$ is defined as: $(\linf[k])_l=\linf_{k+l}$.
For example, if $\linf$ is a real vector space, then the degree-shifted vector space $\linf[1]$ consists of vectors degree $-1$ while coordinates on the associated graded manifold $\manM_\linf$ are of degree $1$.

Let $\manM$ be a $Z$-graded supermanifold. We use the Koszul sign convention, i.e. the graded commutativity reads as $ab=(-1)^{|a||b|}ba$.  Vector fields on $\manM$ act on functions as left derivations, i.e. $V(fg)=(Vf)g+(-1)^{|f|}f(Vg)$. Let $\psi^A$ be a local coordinates on $\manM$ then vector filed $V$ can be written as $V=V^A(\psi)\dl{\psi^A}$.

It is convenient to identify differential form on $\manM$ as functions on the odd tangent bundle $T[1]\manM$, which are polynomial in fiber coordinates. Introducing coordinates $(\psi^A,d\psi^A)$, where $\psi^A$ are base coordinates,the coordinate expression for the de Rham differential reads as
\begin{equation}\label{A1}
    df(\psi)=d\psi^A\frac{\partial f}{\partial \psi^A},
\end{equation}
where $\frac{\partial}{\partial \psi^A}$ denotes the left derivative.

The components components of a differential form are introduced as follows:
\begin{equation}\label{A5}
    \alpha(\psi,d\psi)=\frac{1}{n!}d\psi^{A_n}\dots d\psi^{A_1}\alpha_{A_1\dots A_n}(\psi).
\end{equation}
It follows the graded antisymmetry property of components is determined by $d\psi^Ad\psi^B=(-1)^{(|A|+1)(|B|+1)}\\d\psi^Bd\psi^A$. For example, the components of $2$-form have the following symmetry property:
\begin{equation}\label{A6}
    \omega_{AB}=(-1)^{(|A|+1)(|B|+1)}\omega_{BA}.
\end{equation}

The coordinate expression for the interior product reads as
\begin{equation}\label{A2}
    i_V=V^A(\psi)\frac{\partial}{\partial d\psi^A},
\end{equation}
where $\frac{\partial}{\partial d\psi^B}d\psi^A=\delta^A_B$. It is clear that $Vf=i_V df$ for any function $f\in C^{\infty}(\mathcal{M})$. Note that $|i_V|=|V|-1$. 
The components of $\alpha$ can be expressed as follows
\begin{equation}\label{A3}
    \alpha_{A_1\dots A_n}=i_{A_1}\dots i_{A_n}\alpha(\psi,d\psi), \quad i_{A}=\frac{\partial}{\partial d\psi^A}.
\end{equation}
The Lie derivative is defined by
\begin{equation}\label{A4}
    L_V=i_Vd+(-1)^{|V|}di_V\,.
\end{equation}
Note that $|L_V|=|V|$.

\setlength{\itemsep}{5pt}
\small


\providecommand{\href}[2]{#2}\begingroup\raggedright\endgroup

\end{document}